\documentclass[10pt,journal,compsoc]{IEEEtran}
\usepackage[fleqn]{amsmath}
\usepackage{amsthm}
\usepackage{float}
\usepackage{lipsum}
\usepackage{overpic}
\usepackage{extarrows}
\usepackage{graphicx}
\usepackage{amsfonts}
\usepackage{url}
\usepackage{enumitem}
\usepackage{array}
\usepackage{hyperref}
\usepackage{booktabs}
\usepackage{multirow}
\usepackage{tikz}
\DeclareMathOperator*{\argmax}{arg\,max}
\newtheorem{theorem}{Theorem}
\newtheorem{lemma}{Lemma}

\ifCLASSOPTIONcompsoc
    \usepackage[caption=false, font=normalsize, labelfont=sf, textfont=sf]{subfig}
\else
\usepackage[caption=false, font=footnotesize]{subfig}
\fi
\usepackage{algorithm}
\usepackage{algpseudocode}
\usepackage{amsmath}
\begin{document}
%
\title{Correlation-Based Community Detection}
%
%
%
%

\author{Zheng Chen,
        Zengyou He,
        Hao Liang, Can Zhao and Yan Liu
\IEEEcompsocitemizethanks{
\IEEEcompsocthanksitem Z. Chen, H. Liang and Y. Liu are with School of Software, Dalian University of Technology, Dalian, China.
\IEEEcompsocthanksitem Z. He is with School of Software, Dalian University of Technology, Dalian,China, and Key Laboratory for Ubiquitous Network and Service Software of Liaoning Province, Dalian, China.\protect\\
E-mail: zyhe@dlut.edu.cn

\IEEEcompsocthanksitem C. Zhao is with Institute of Information Engineering, CAS.  }
}

%
%

\markboth{Journal of \LaTeX\ Class Files,~Vol.~14, No.~8, August~2015}%
{Shell \MakeLowercase{\textit{et al.}}: Bare Demo of IEEEtran.cls for Computer Society Journals}
%



\IEEEtitleabstractindextext{%
\begin{abstract}
 Mining community structures from the complex network is an important problem across a variety of fields. Many existing community detection methods detect communities through optimizing a community evaluation function.  However, most of these functions even have high values on random graphs and  may fail to detect small communities in the large-scale network (the so-called resolution limit problem). In this paper, we introduce two novel node-centric community evaluation  functions by connecting correlation analysis with community detection. We will further show that  the correlation analysis can provide a novel theoretical framework which unifies some existing evaluation functions  in the context of a correlation-based optimization problem.  In this framework, we can mitigate the resolution limit problem and eliminate the influence of random fluctuations by selecting the right correlation function. Furthermore,  we introduce three key properties used in mining association rule into the context of community detection to help us choose the appropriate correlation function. Based on our introduced correlation functions, we propose a  community detection algorithm called CBCD. Our proposed algorithm outperforms existing state-of-the-art algorithms on both synthetic benchmark networks and real-world networks.
\end{abstract}

}

\maketitle

\IEEEdisplaynontitleabstractindextext

%
\IEEEpeerreviewmaketitle

\IEEEraisesectionheading{\section{Introduction}\label{sec:introduction}}

%
%
%
%
\IEEEPARstart{C}{OMMUNITY} detection plays a key role in network science, bioinformatics~\cite{spirin2003protein}, sociological analysis~\cite{papadopoulos2012community} and data mining. It not only helps us identify the network modules, but also offers insight into how the entire network is organized by local structures. The detected communities could be interpreted as the basic modules of various kind of networks, e.g. social circles in social networks~\cite{feld1981focused}, protein complexes in protein interaction network~\cite{krogan2006global}, or groups of organisms in food web network~\cite{girvan2002community}. More generally, a widely accepted consensus on community~\cite{Fortunato2009Community} is that the community should be a set of vertices that has more edges within the community than edges linking vertices of the community with the rest of the graph.


Although community detection has been extensively investigated during the past decades, there is still no common agreement on a formal definition regarding what a community exactly is. Many existing community detection algorithms are based on the previously mentioned criterion (more internal connections than external connections), with proposed quality metrics that quantify how community nodes connect internal nodes densely and external nodes sparsely. For example, popular metrics such as \textit{betweenness}~\cite{girvan2002community}, \textit{modularity}~\cite{newman2004finding}, \textit{conductance}~\cite{Chung1997Spectral}, \textit{ratio cut}~\cite{wei1989towards}, \textit{density}~\cite{mancoridis1998using} and \textit{normalized cut}~\cite{shi2000normalized} are all based on this intuitive idea. And existing related algorithms first utilize these metrics to derive globally-defined objective functions, then maximize (or minimize) the objective function by partitioning the whole graph  \cite{newman2004fast} \cite{schuetz2008efficient} \cite{pujol2006clustering} \cite{noack2009multi} \cite{spielman2004nearly} \cite{andersen2006communities} \cite{andersen2006local} \cite{kloster2014heat} \cite{zhu2013local}  \cite{van2016local} \cite{benson2016higher}.

The traditional view on community evaluation  relies on counting edges in different ways. Simply doing that is not a sensible way, and it is not sufficient to convince people that identified community structure is pronounced. In \cite{fortunato2016community},
\begin{figure}[!ht]
    \centering
	  \subfloat[]{
       \includegraphics[width=0.9\linewidth]{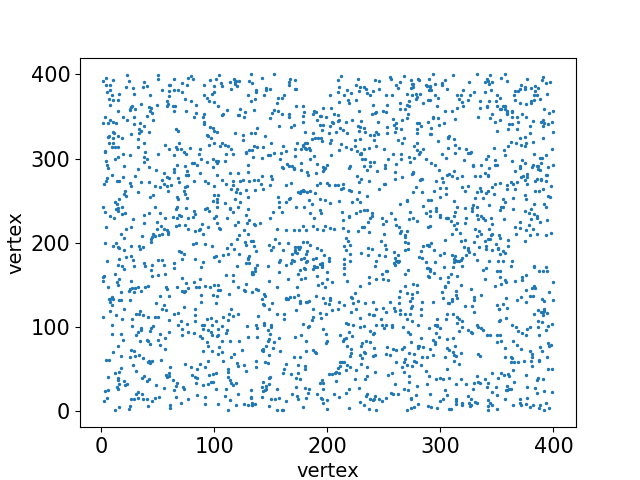}}
    \label{1a}\hfill
	  \subfloat[]{
        \includegraphics[width=0.9\linewidth]{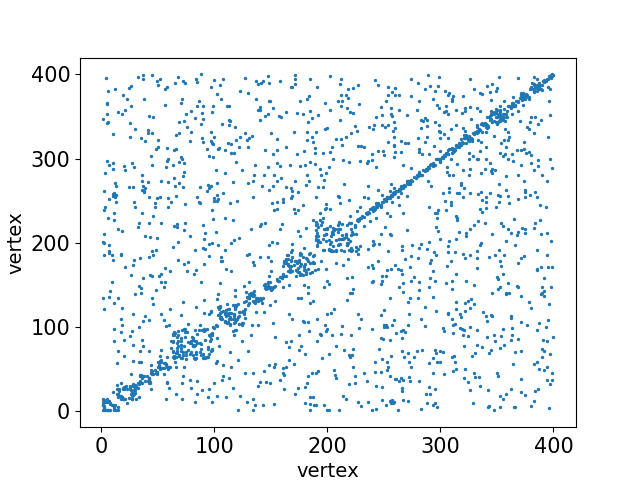}}
    \label{1b}\\
	  \caption{ Two matrices of the same random graph that is generated according to the E-R model. Fig. 1(a) shows the original matrix. Fig. 1(b) is a matrix obtained by reshuffling the order of vertices.}
	  \label{fig2}
\end{figure}
Fortunato and Hric pointed  out an amazing fact that an Erdos-Renyi (E-R) random graph~\cite{erdos1959random} can generate the modular structures. To exemplify this phenomenon, a 400$\times$400 matrix of E-R random graph is illustrated in Fig. 1. As shown in Fig. 1(a), the matrix is evidently disordered and random such that no one believes there exist community structures.
However, as long as rearranging the entries of the matrix  by reorganizing the index of the vertices, the modular structures emerge at the diagonal of the matrix. Such structures  should not be real and are generated  due to the random fluctuations in the network construction process~\cite{fortunato2016community}. This weird phenomenon causes that many clustering algorithms whose metrics do not consider the influence of random fluctuations  identify communities even in the random networks as well. From the above illustrations, it makes sense that community metrics should  be considered how significant their result is.

Probably the most popular community quality metric is the modularity introduced by Newman and Girvan~\cite{newman2004finding}. This metric is based on a prior work about a measure of assortative mixing which is proposed in~\cite{newman2003mixing}. It evaluates the quality of a partition of the network. The global expression of the modularity function is:
$$ Q = \frac{1}{2m}\sum_{ij}(W_{ij}-E_{ij})\delta(C_i,C_j),        \eqno{(1)}$$
where $m$ is the number of edges, $W_{ij}$ is the element of the adjacency matrix $W$, $\delta(x,y)$ is the kronecker delta function whose value is 1 if $x=y$ and 0 otherwise, $C_i$ and $C_j$ represent the community index of  $i$  and  $j$  respectively, $E_{ij}$ is the expected number of edges that connect vertex $i$ and vertex $j$ under the random graph \textit{null model}. The null model specifies how to generate a random graph that preserves some characteristics of the original graph. The most commonly used null model is the configuration model~\cite{bollobas1980probabilistic}~\cite{molloy1995critical}, where the degrees of all vertexes are preserved in the random graph. Under the configuration model, one can derive that $E_{ij} = d_id_j/2m$ ($d_i$ and $d_j$ are the degree of $i$ and $j$, respectively). Now we can reformulate the modularity as:
$$ Q = \sum_{j} \Big[ \frac{l_j}{m}- \Big( \frac{K_j}{2m} \Big)^2 \Big],        \eqno{(2)}$$
where $K_j$ represents the sum of degrees of all vertexes in community $j$ and $l_j$ is the number of internal edges within community $j$. At the same time, we could use $\l_j/{m}-(K_j/2m)^2$ as a metric for evaluating the quality of community $j$. The intuition behind this formula is that simply counting edges are not sufficient to determine a true community structure so that the number of expected edges in the null model should be incorporated as well. Unfortunately,  the modularity maximum does not equal to the most pronounced community structure.  This is the well-known resolution limit problem which modularity suffers from, i.e.,  the modularity function may fail to detect modules which are smaller than a scale in large networks.  Many techniques have been used to mitigate the resolution limit problem \cite{Fortunato2007Resolution}, such as the multi-resolution method.  Note that the multi-resolution method does not provide a reliable solution to the problem \cite{Lancichinetti2011Limits}.  In addition, it is also known that the modularity can be a high value even in the E-R random graph \cite{guimera2004modularity}.  This seems counterintuitive, but it is the fact since the modularity is a kind of measurement which measures the distance between real network structure and the ``average" of random network structures. As a result, there is no sufficient information  to confirm the distribution of the modularity.

In addition,  the community is a local structure of  the network. We should examine the community from the local view.
The strength of connectedness between one node and a community embeds the local feature of that community. Using these local features to construct a global evaluation function  for a community helps us avoid from missing the local structural  information.  However, how to assess the strength of connectedness between one node and the community is a challenging task since there is no convinced definition on what the strength of connectedness is.

\textit{Our Contributions}.  In this paper,  we propose a correlation-based community detection framework, in which the local connectedness strength between each node and a community is assessed through correlation analysis. More precisely, we represent the basic structural information of a node as the corresponding row vector embedded in the adjacency matrix and encode the structural information of a community into a binary vector. To demonstrate the feasibility and advantages of this framework,  we introduce two concrete node-centric community metrics  based on the correlation between a node and a community, PS-metric (node version of modularity) and $\phi$-Coefficient metric.  Moreover, we provide the detailed theoretical analysis for our metrics, which show that these metrics are capable of mitigating the resolution limit problem and alleviating the fake community issue in E-R random graph.    Besides, we present a correlation-based community detection (CBCD) method which adopts  two introduced metrics to identify the community structure. Experimental results on both real networks and the LFR networks show that CBCD outperforms state-of-the-art methods. The summary of contributions of this paper is listed as follows:
\begin{itemize}
\item To the best of our knowledge, we are the first to introduce correlation analysis into the node-centric metrics, which calculate the correlation value between a node and a community. Besides,  we slightly modify three key properties for mining association rule \cite{PiatetskyShapiro1991DiscoveryAA} in the context of community detection, which can guide us select the right correlation function for community detection.

\item We further  investigate the relationship between the correlation analysis and community detection.  It has been found that correlation analysis can be viewed as a theoretical interpretation framework for community detection, which unifies some exiting metrics and provide the potential of the deriving new and better community evaluation measures.

\item The introduced PS-metric (node version of modularity) is less affected by the random fluctuations. We not only give the detailed theoretical proof, but also show the empirical results to validate the effectiveness of this metric under the E-R model. Moreover, we prove that the $\phi$-Coefficient metric can mitigate the resolution limit problem and give an intuitive interpretation.

\item Experimental results demonstrate that our method outperforms state-of-the-art methods on  real networks and LFR networks.
\end{itemize}

The organization of this paper is structured as follows: Section 2 discusses the related work.  Section 3 introduces several basic definitions of correlation analysis and proves some properties of proposed metrics.  Section 3 describes our corresponding community detection method. Section 4 presents the experimental results of CBCD along with the other methods. Section 5 concludes this paper.



\section{Related Work}
\subsection{Correlation-Based Method}
There are already some algorithms in the literature that utilize the different types of correlation information hidden behind the network structure to solve the community detection problem. Once the correlation definition is specified, community detection can be cast as an optimization problem by maximizing a correlation-based objective function. The existing community detection algorithms based on correlation can be categorized according to the different correlation definitions.
\subsubsection{Node-Node Similarity}
The pairwise similarity between two nodes can be viewed as a kind of correlation. The goal of community detection based on node-node similarity is to put the nodes which are close to each other into the same group.  In \cite{pons2005computing}, a random-walk based node similarity was proposed, which can be used in an agglomerative algorithm to efficiently detect the communities in the network. The method in \cite{pan2010detecting} is based on the node similarity proposed in~\cite{zhou2009predicting} to find communities in an iterative manner.
\subsubsection{Node-Node Correlation}
A covariance matrix of the network is derived from  the incidence matrix in \cite{shen2010covariance}, which can be viewed as the unbiased version of the well-known modularity matrix. A correlation matrix is obtained through introducing the re-scaling transformation into the covariance matrix, which significantly outperforms the covariance matrix on the identification of communities.  The algorithm in \cite{zarei2009detecting} constructs a node-node correlation matrix based on the Laplacian matrix so as to incorporate the feature of NMF (non-negative matrix factorization) method. In \cite{PhysRevX.5.021006}, MacMahon et al. introduce the appropriate correlation-based counterparts of the most popular community detection techniques via a consistent redefinition of null models based on random matrix theory. A correlation clustering \cite{bansal2004correlation} based community detection framework is proposed in \cite{veldt2018correlation}, which unifies the modularity, normalized cut, sparsest cut, correlation clustering and cluster deletion by introducing a single resolution parameter $\lambda$.

\subsubsection{Edge-Community  Correlation}
In~\cite{duan2014community}~\cite{duan2017utilizing}, Duan et al. connect the modularity with correlation analysis by reformulating modularity's objective function as a correlation function based on the probability that an edge falls into the community under the configuration model. Hence, it can be viewed as a method based on  edge-community correlation information.

\subsection{Node-Centric Method}
Almost all classical community detection metrics implicitly assume that all nodes in a community are equally important. These metrics, viewing the community  as a whole, only focus on the total number of edges  within the community, the sum of degree of all nodes or the size of the community. The connection strength between  each vertex and the community is not involved in these metrics. This  may result in missing important structural  information so that  two distinct communities cannot be distinguished  in certain  circumstances.  Node-centric methods take into account how each vertex connects the community densely. This is consistent with the fact that the community is a local structure of the network.

Chakraborty et al. \cite{chakraborty2014permanence} \cite{chakraborty2016genperm} \cite{chakraborty2016permanence} propose a new node-centric metric called permanence, which describes the membership strength  of a node to a community. The central idea behind permanence is to consider  the following   two factors: (1) the distribution of the external edges rather than the number of all external connections (2) the strength of the internal-connectivity instead of the number of all  internal edges. The corresponding algorithm is to maximize the permanence-based objective function. WCC is another node-centric metric proposed in~\cite{prat2014high}~\cite{prat2016put}, which considers the triangle as the basic structure instead of the edge or node. WCC of a node consists of two parts: isolation and intra-connectivity. The proposed algorithm aims at the optimization of a WCC-based objective function.  Focs \cite{bandyopadhyay2015focs} is a heuristic method that accounts for local connectedness.  The local connectedness of a node to a community depends upon two scores of the node with respect to the community:  community connectedness and neighborhood connectedness. Focs mainly consists of two phases: leave phase and expand phase, which is respectively based on community connectedness and neighborhood connectedness.

\subsection{Summary}
Our method is different from all above existing  methods. Although the proposed metrics in our method are node-centric as well, they are derived from the  perspective of correlation analysis. Different from the existing correlation-based methods, our method utilizes the correlation  between the node and the community based on a $2\times2$ contingency table. Thus, our method can be viwed as a node-community correlation based method. Besides, we will choose appropriate correlation functions to mitigate the resolution limit problem under the guidance of~\cite{tan2004selecting}~\cite{duan2014selecting}. The further analysis will be given  in section 3.

\section{Correlation Analysis In Network}

In this section, we first formalize the community detection problem, then, introduce two correlation measures and extend it for community structure evaluation. Two new correlation metrics for community detection will be proposed in this section. Next we will describe modularity from correlation analysis perspective and discuss its relation with our metrics. At last, we will discuss the desirable properties for new correlation metric.
\subsection{Preliminaries}
\noindent \textbf{Undirected Graph.} Let $G = (V,E,W)$ be an undirected graph with $n$ nodes and $m$ edges, where $V$ is the node set, $E$ is the edge set, $W$ is the adjacency matrix. If an edge $(u,v)\in E$, then $W_{uv} = 1$ and $W_{uv}=0$ otherwise. In the case of undirected graph, $W_{uv} = W_{vu}$, it means that the adjacency matrix $W$ is a symmetric matrix.

\noindent \textbf{Community Detection.} Given a graph $G = (V,E,W)$, the goal of community detection is to partition the graph with $|V| = n$ vertices into $l$ pairwise disjoint groups $P = \{ S_{1},...,S_{l} \}$, where $S_{1}\bigcup..\bigcup S_{l} = V$ and $S_{i}\bigcap S_{j} =  \emptyset$ for any $i,j$.

To begin with, we define $f(u,S)$ as a measurement function which takes a vertex $u$ and a community $S$ containing $u$ as the input, and return a real value which indicates the connectivity of vertex $u$ regarding the community $S$. The function $f(u,S)$ should satisfy certain basic requirements so that it can be used to evaluate the structure of community $S$. First, $f(u,S)$ must be bounded to guarantee the convergence of community search algorithms. Second, the more strongly a vertex connects to a community, the higher $f(u,S)$ becomes. The detailed properties that a measurement function needs to possess for satisfying the second requirement will be further discussed in section 3.2. Now, we can define the vertex-centric metric of a community $S$ as the sum of the function $f(u,S)$ of all members $u$ that belong to the community $S$:
$$ F(S) = \sum_{u\in S}f(u,S).        \eqno{(3)}$$
Analogously to what we have done before, we define the objective function of a partition  $P$ through taking the sum of the vertex-centric metric value of each community $S_{i} \in P$:
$$ \Gamma (P) = \sum_{j=1}^{l} F(S_{j}).        \eqno{(4)}$$
Given a graph $G = (V,E,W)$, community detection can be cast as an optimization problem with (4) as the objective function. The partition $P$ is optimal when $\Gamma(P)$ achieves a maximum value, and we call this partition the optimal partition.

Now the key point lies in how to formulate a feasible $f(u,S)$. A natural idea is to use the correlation function between $u$ and $S$ as $f(u,S)$. The correlation function measures the correlation relationship between $u$ and $S$, in which a higher correlation value indicates that there is a strong affinity between $u$ and $S$. However, we need to choose an appropriate correlation model to represent the connectivity cohesion of vertex $u$ regarding community $S$. Moreover, the information about the topology structure of a node and a community should also be considered in this model. In this paper, we convert the set of nodes which belong to a community into a binary vector to embody the basic structural information of the community. The positions of nodes in the graph are determined by their neighbourhoods, that is, if two nodes connect to the same set of neighbors, then the whole graph will retain the same structure after exchanging the positions of these two nodes. Therefore, the local structure of a node can be depicted by its row vector embedded in the adjacency matrix. In summary, we have the following definitions.

\noindent
\textbf{Community vector and Neighbor vector.} The community vector of $u$ with respect to $S$, $\Psi_{S\backslash \{ u \} }=(e_1, ... e_{u-1},e_{u+1}, ...e_n)$, is a binary vector of length $n-1$, where $e_v = 1$ if vertex $v$ belongs to community $S$ and $e_v = 0$ otherwise. $\Psi_{S\backslash \{ u \}}[v]$ denotes the value of vector element $e_v$. $\psi_{u}=(g_1, ... g_{u-1},g_{u+1}, ...g_n)$ is called the neighbor vector of vertex $u$, where $g_{v} = 1$ if there exists an edge between $u$ and $v$, and $g_{v} = 0$ otherwise. $\psi_{u}[v]$ denotes the value of vector element $g_v$. Note that vertex $u$ is excluded from the vector, we will give an explanation from the probabilistic perspective in the following part.

The above two binary vectors can be viewed as the samples of two binary variables $\mathcal{C}(x,u)$ and $\mathcal{G}(x,S\backslash \{ u \})$, where $\mathcal{C}(x,u) = \psi_{u}[x]$ and $\mathcal{G}(x,S\backslash \{ u \})=\Psi_{S\backslash \{ u \}}[x]$. A $2\times2$ contingency table for these two variables is given in Table 2. The entry $f_{11}$ is denoted by $\omega$, that represents the count of inner connections of vertex $u$ regarding community $S$. The entry $f_{11}+f_{10}$ is denoted by $\epsilon$, $\epsilon = |S|-1$ and $|S|$ is the size of community $S$.
\begin{table}[!h]
\renewcommand{\arraystretch}{1.7}
\caption{
A $2\times2$ contingency table for vertex $u$ and community $S$ .}
\label{table3}
\centering
\begin{tabular}{|c|c|c|c|}
\hline
 & $\mathcal{C}(x,u)=1$ & $\mathcal{C}(x,u) = 0$ & $ \sum_{j}f_{ij}$  \\
\hline
$\mathcal{G}(x,S\backslash \{ u \})=1$  & $\omega$  & $f_{10}$  & $\epsilon$  \\
\hline
$\mathcal{G}(x,S\backslash \{ u \})=0$  & $f_{01}$  & $f_{00}$  & $f_{01}+f_{00}$   \\
\hline
$ \sum_{i}f_{ij}$   & $d_{u}$  & $f_{10}+f_{00}$  & $N$  \\
\hline
\end{tabular}
\end{table}
The entry $f_{11}+f_{01}$ is denoted by $d_{u}$, where $d_{u}$ is the degree of vertex $u$. The uppercase Roman letter $N$ is the length of the binary vector, where $N=n-1$.

\begin{figure*}[!htb]
    \centering
		\begin{overpic}[scale=0.7]{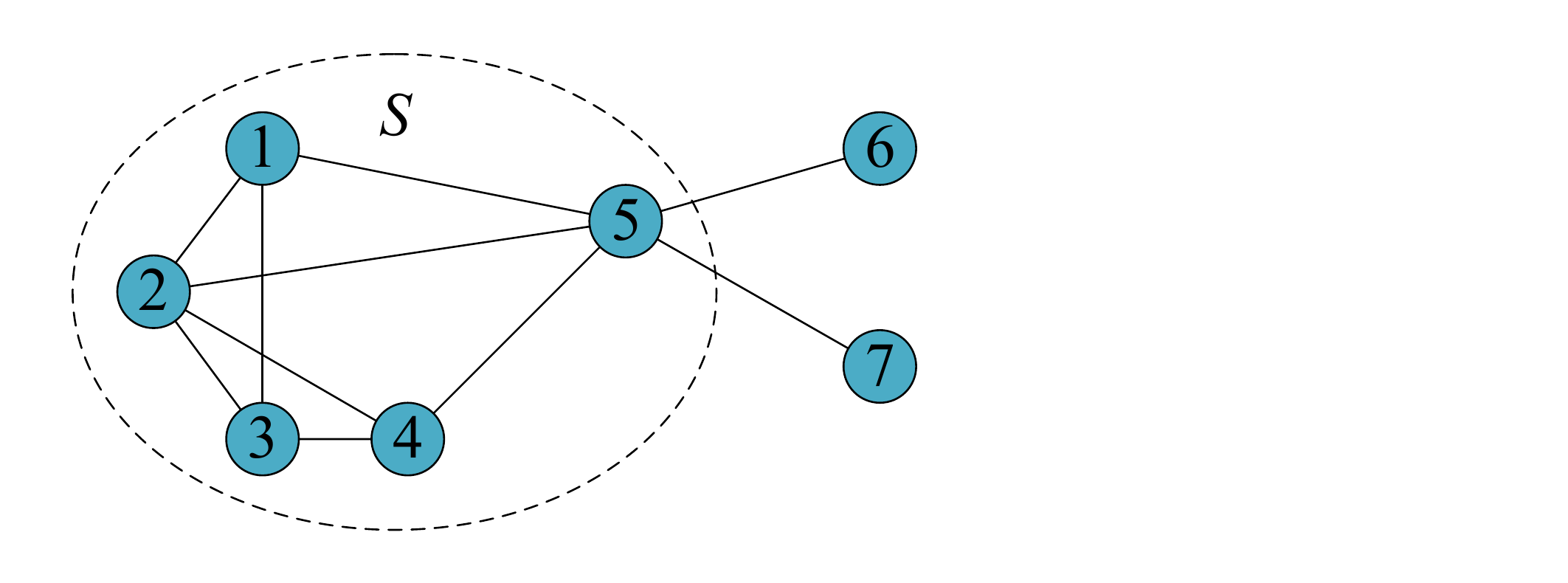}
			\put(70,20){\scalebox{1.5}{$\Psi_{S\backslash \{ 5 \} } = (e_1,e_2,e_3,e_4,e_6,e_7)$}}
			\put(81.5,15){\scalebox{1.5}{$=(1,1,1,1,0,0)$}}
			\put(76.9,10){\scalebox{1.5}{$\psi_{5} = (g_1,g_2,g_3,g_4,g_6,g_7)$}}
			\put(81.9,5){\scalebox{1.5}{$=(1,1,0,1,1,1)$}}
		\end{overpic}
	 \caption{ One example graph with 7 nodes and 10 edges, and $S$ is a community of size 5. The community vector $\Psi_{S\backslash \{ 5 \} }$ and neighbor vector  $\psi_{5}$ are given in the right part of the figure. }
	  \label{fig2}
\end{figure*}
 Now we can give several probability definitions that will be used in the correlation measure. Firstly, to simplify the notations, $\mathcal{C}$ is used to denote $\mathcal{C}(x,u)=1$, and $\mathcal{G}$ is used to denote $\mathcal{G}(x,S\backslash \{ u \})=1$. Then, $P(\mathcal{C} \mathcal{G}) = \frac{\omega}{N}$ is the probability that a randomly chosen vertex from the vertex set $V\backslash \{ u \}$ connects to vertex $u$ and belongs to community $S$ simultaneously. $P(\mathcal{C}) = \frac{d_{u}}{N}$ is the probability that a randomly selected vertex from $V\backslash \{ u \}$ has a connection with $u$. $P(\mathcal{G}) = \frac{\epsilon}{N}$ is the probability that a randomly chosen vertex from $V\backslash \{ u \}$ belongs to community $S$. Under the assumption of independence, the probability of $\mathcal{C} \mathcal{G}$ can be calculated by $P(\mathcal{C} \mathcal{G}) = P(\mathcal{C})P(\mathcal{G})={\epsilon d_{u}}/{N^2}$.

The above definitions and notations are exemplified in Fig. 2, where there is a community $S$ of size 5. Here we let $u=5$, $\mathcal{C}$ denotes $\mathcal{C}(x,5)=1$ and $\mathcal{G}$ denotes $\mathcal{G}(x,S\backslash \{ 5 \})=1$. We will calculate $P(\mathcal{G})$, $P(\mathcal{C})$ and $P(\mathcal{C}\mathcal{G})$. To begin with, we should know how to get two vectors $\Psi_{S\backslash \{ 5 \} }$ and $\psi_{5}$. In Fig. 2, since node 5 connects almost all the nodes except node 3, $\psi_{5}$ has only one zero entry $g_3$. Since node 6 and node 7 are not included in community $S$ , their corresponding entries $e_{6}$ and $e_{7}$ are both zeros. After getting these two binary vectors, $\omega$, $d_5$, $N$ and $\epsilon$  can be calculated accordingly, that is, $\omega=3$, $d_5=5$, $N=6$, $\epsilon=4$. Then, we have:
$$ P(\mathcal{G}) = \frac{4}{6}, P(\mathcal{C}) = \frac{5}{6}, P(\mathcal{C}\mathcal{G}) = \frac{3}{6}.       $$

\subsection{Correlation Measure}
Many functions have been proposed to measure the correlation between two binary vectors
 in statistics, data mining and machine learning. There are some guidelines~\cite{tan2004selecting}~\cite{geng2006interestingness}~\cite{duan2014selecting} provided for users to select the right measure according to their needs. The most measures consist of $P(\mathcal{G})$, $P(\mathcal{C})$ and $P(\mathcal{C}\mathcal{G})$. Here we will list several properties to instruct us in choosing the appropriate measure for community detection. Let $M$ be a measure for correlation analysis between two variables $\mathcal{C}$ and $\mathcal{G}$. In~\cite{PiatetskyShapiro1991DiscoveryAA}, Piatetsky-Shaprio came up with three key properties about a good correlation measure for association analysis:

\vspace{0.15cm}
\noindent \textbf{P1}: $M = 0$ if $\mathcal{C}$ and $\mathcal{G}$ are statistically independent;

\noindent \textbf{P2}: $M$ monotonically increases with $P(\mathcal{C}\mathcal{G})$ when $P(\mathcal{C})$ and $P(\mathcal{G})$ are both fixed;

\noindent \textbf{P3}: $M$ monotonically decreases with $P(\mathcal{C})$ (or $P(\mathcal{G})$) when $P(\mathcal{C}\mathcal{G})$ and $P(\mathcal{G})$ (or  $P(\mathcal{C}\mathcal{G})$ and $P(\mathcal{C})$) are fixed.
\vspace{0.15cm}

 Note that these three properties are used in association analysis, we need further investigation in the context of community detection. \textbf{P1} indicates that $M$ should be able to measure the deviation from statistical independence. Then, the higher $M$ is, the stronger dependence between
 $\mathcal{C}$ and $\mathcal{G}$ is. Node $u$ can be regarded as the member of community $S$ when $M>0$. In this paper, we use the E-R model as the underlying random graph model to describe the statistical independence. Given the graph $G$ with $n$ nodes and $m$ edges, one random graph under the E-R model is generated by forming an edge between any two nodes randomly and independently with the probability $p=2m/(n(n-1))$. By redefining \textbf{P1}, $M$ should be a measure such that the expected value of $M$ equals to zero for any node $u$ and community $S$ under the E-R model.
\begin{figure}[!ht]
    \centering
       \includegraphics[width=0.9\linewidth]{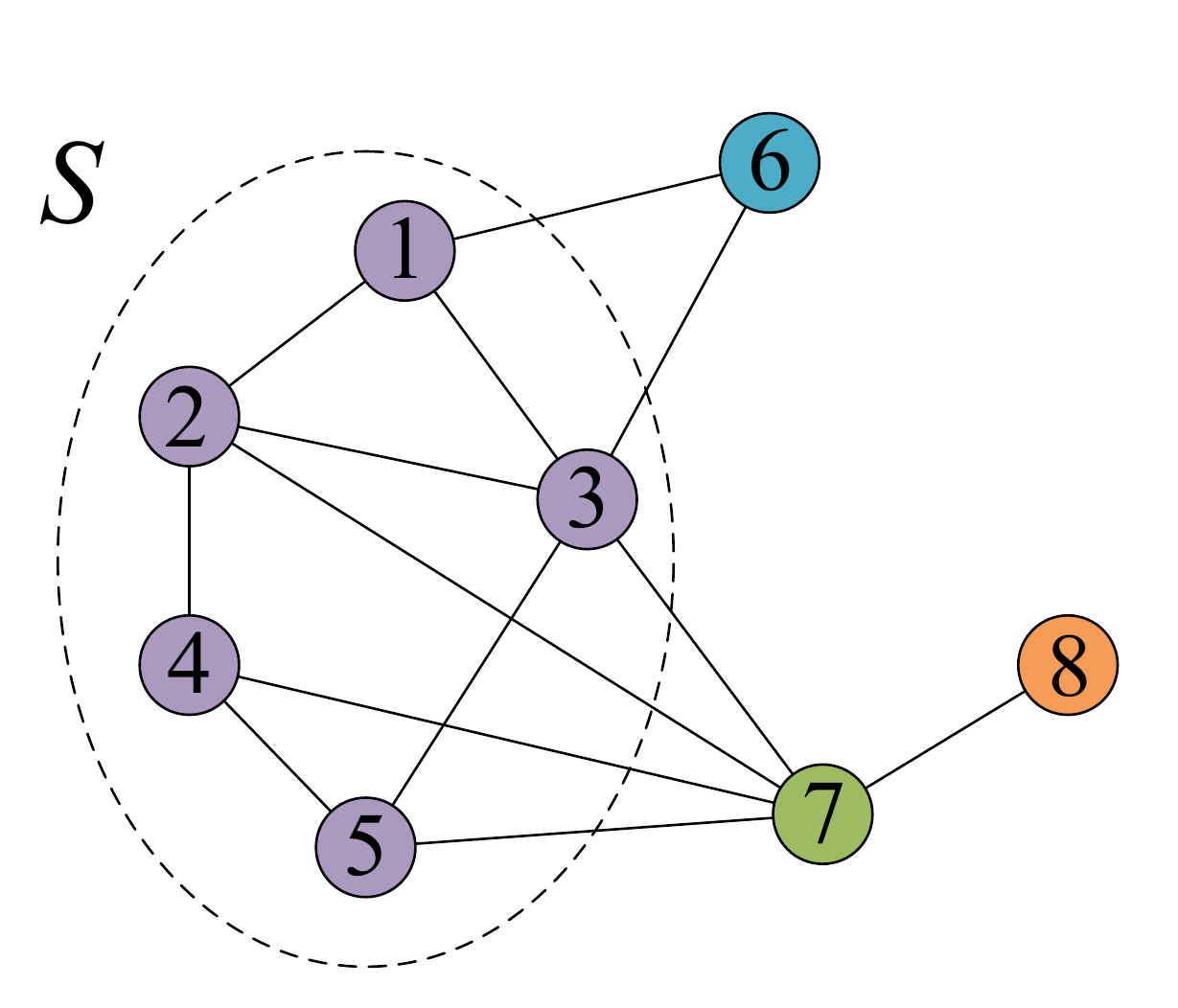}
	  \caption{ Node 7 has stronger intra-connection than node 6.}
	  \label{fig3}
\end{figure}

 According to the definitions introduced in section 3.1, $P(\mathcal{C}\mathcal{G})$, $P(\mathcal{C})$ and $P(\mathcal{G})$ are respectively decided by $\omega$, $d_u$ and $\epsilon$ since $N$ is fixed. As a result, the monotonicity of $P(\mathcal{C}\mathcal{G})$, $P(\mathcal{C})$ and $P(\mathcal{G})$ is respectively determined by $\omega$, $d_u$ and $\epsilon$ as well. \textbf{P2} and \textbf{P3} describe two features of the cohesion of $u$ with respect to community $S$: intra-connection and isolation. The intra-connection of node $u$ with respect to community $S$ indicates how node $u$ connects community $S$ densely, and it can be quantified by $\omega/\epsilon$. This is exemplified in Fig. 3, where node 7 connects four nodes of subgraph $S$ and node 6 connects two nodes of subgraph $S$. Despite node 6 connects subgraph $S$ with all its edges, node 7 has more links than node 6 within subgraph $S$.  Apparently, node 7 is more qualified for the member of community $S$. Note that the increase of $\omega$ will lead to the increase of  $\omega/\epsilon$ when $\epsilon$ and $d_u$ are both fixed. According to \textbf{P2}, $M$ should increase with the increment of  $\omega/\epsilon$, which means that strength of the intra-connection between $u$ and $S$ is becoming higher.
\begin{figure}[!ht]
    \centering
       \includegraphics[width=0.9\linewidth]{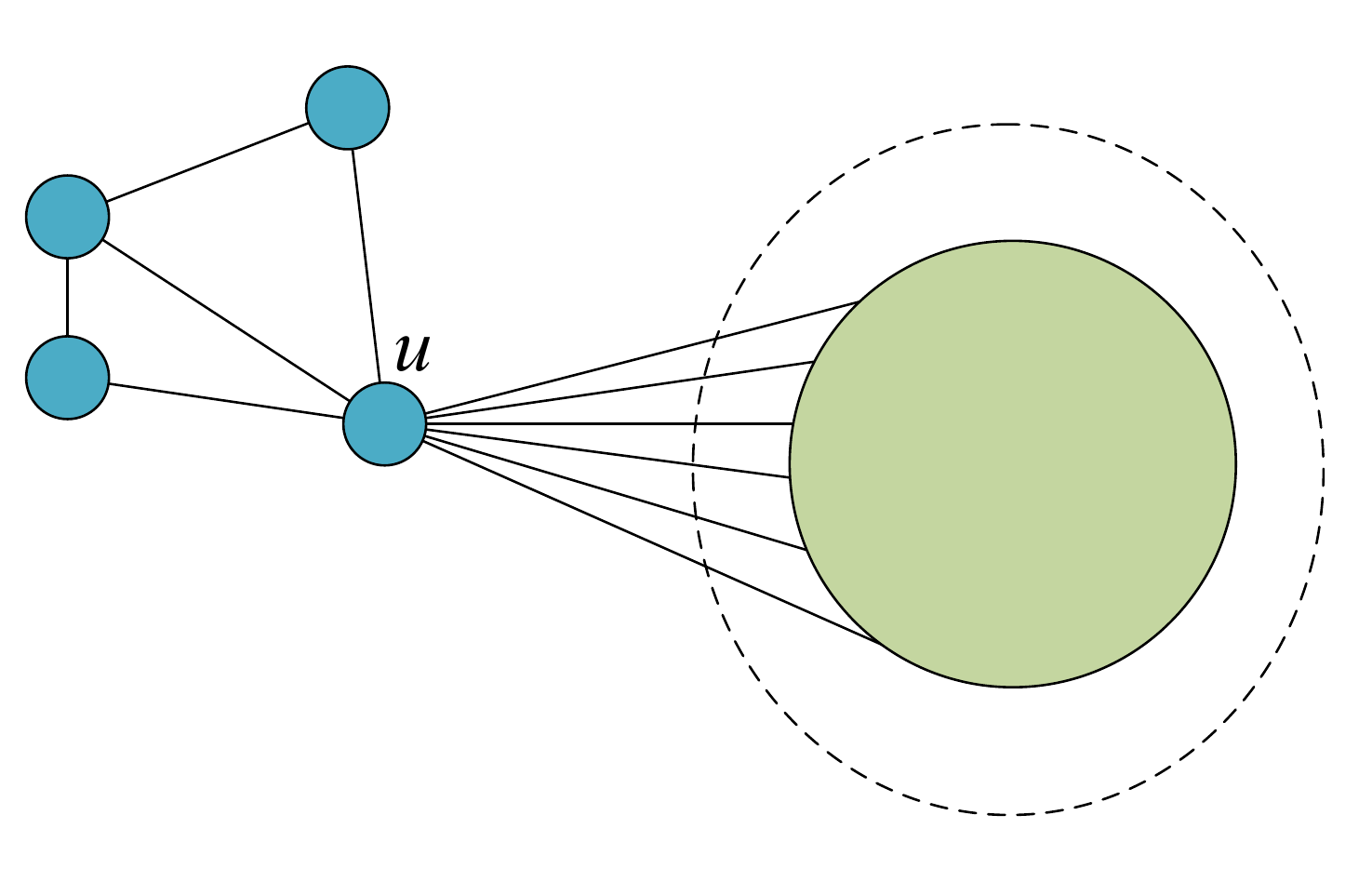}
	\centering	
 	\caption{ Node $u$ does not have strong isolation from the rest of graph which is circled by dash line.}
	  \label{fig4}
\end{figure}
 In a similar way for \textbf{P3}, $M$ should monotonically decrease with the increment of $\epsilon$ when $\omega$ and $d_u$ are both fixed since $\omega/\epsilon$ will decrease. On the other hand, it will lead to the biased result if we exclusively maximize the intra-connection. In Fig. 4, node $u$ have connections with all the left three nodes, but we cannot confidently infer that node $u$ and the left three nodes should be put together to form a community. This is because the number of links from $u$ connecting the rest of the graph is twice its number of links with the left three nodes. A node belonging to a community should have strong isolation from the rest of the graph. The isolation can be quantified by $\omega/d_u$, where the higher $\omega/d_u$ indicates that node $u$ connects external nodes sparsely. It is easy to see that the increment of $\omega$ will lead to the increment of $\omega/d_u$  when $\epsilon$ and $d_u$ are both fixed. In regard to \textbf{P2}, $M$ should increase with the increment of $\omega/d_u$, which means that the isolation of $u$ from the rest of the graph is becoming stronger. In a similar way for \textbf{P3}, $M$ should monotonically decrease with the  increment of $d_u$ when $\omega$ and $\epsilon$ are both fixed since $\omega/\epsilon$ will decrease.

To further analyze the relationship between correlation measure and community detection, let us consider the  confidence measure of $\mathcal{C}$ and $\mathcal{G}$. We have:
$$\textit{confidence}( \mathcal{C}\Rightarrow\mathcal{G} ) = P( \mathcal{G}|\mathcal{C}) = \frac{P(\mathcal{G}\mathcal{C})}{P(\mathcal{C})} = \frac{\omega}{d},$$
$$\textit{confidence}( \mathcal{G}\Rightarrow\mathcal{C} ) = P( \mathcal{C}|\mathcal{G}) = \frac{P(\mathcal{G}\mathcal{C})}{P(\mathcal{G})}= \frac{\omega}{\epsilon},$$
where $\textit{confidence}(\mathcal{C} \Rightarrow \mathcal{G})$ and  $\textit{confidence}(\mathcal{G}\Rightarrow\mathcal{C})$ are defined as  the neighborhood connectedness score (nb-score) of a node and  the community connectedness score (com-score) of a node in the Focs algorithm \cite{bandyopadhyay2015focs},  respectively. These two node-based measures are the theoretical basis of the Focs algorithm. This fact tells us that correlation analysis can be appropriately applied in community detection.

\subsubsection{Piatetsky-Shapiro's  Rule-Interest}

Piatetsky-Shapiro's rule-interest~\cite{PiatetskyShapiro1991DiscoveryAA} measures the difference between the true probability and the expected  probability    under the assumption of independence between $\mathcal{C}$ and $\mathcal{G}$. Piatetsky-Shapiro's rule-interest is appropriate for the task of community detection, and we will show that modularity can be explained under our framework when the correlation measure is Piatetsky-Shapiro's rule-interest later on. Our measurement function based on Piatetsky-Shapiro's rule-interest is calculated as:
$$  PS(u,S) = P(\mathcal{C}\mathcal{G}) - P(\mathcal{C})P(\mathcal{G}) = \frac{\omega}{N}-\frac{\epsilon d_u}{N^2}. \eqno{(5)}$$
Recalling the aforementioned variables in Table 2, the degree of vertex $u$ is $d_u$, the size of set $S\backslash \{ u \}$ is $\epsilon$ and the number of links between vertex $u$ and community $S$ is $\omega$. If we randomly select a node from the node set $V\backslash \{ u \}$, the probability that the selected node simultaneously connects  vertex $u$ and belongs to community $S$ is $\frac{\omega}{N}$. Similarly, a randomly selected node from  $V\backslash \{ u \}$ connects vertex $u$ with the probability $\frac{d_u}{N}$ and belongs to community $S$ with the probability  $\frac{\epsilon}{N}$.  In the case of random graph, the expected probability that the chosen node simultaneously connects vertex $u$ and belongs to community $S$ is $\frac{\epsilon}{N} \cdot \frac{d_u}{N}$.

The modularity function has a few variants, but they are all identical in nature. Without loss of generality, we will adopt the definition given in Equation (2) as our modularity function. Given a community $S$, the partial modularity of $S$ is reformulated as:
\begin{equation*}
    \begin{split}
        Q_S &= \frac{l_S}{m}- \Big( \frac{K_S}{2m} \Big)^2 = \frac{l_S}{m} - \frac{\sum_{j \in S}d_j}{2m} \cdot \frac{\sum_{i \in S}d_i}{2m},
    \end{split}
\end{equation*}
where $m$ is the number of the edges in graph $G$ and $l_S$ is the number of the edges inside community $S$. If we randomly choose an edge from the edge  set $E$, the probability of the chosen edge inside community $S$ is ${l_S}/{m}$. Likewise, the probability that one randomly chosen edge has at least one end inside community $S$ is $\frac{\sum_{j \in S}d_j}{2m}$. Then, the expected probability of the chosen edge inside community $S$ is $\frac{\sum_{j \in S}d_j}{2m} \cdot \frac{\sum_{i \in S}d_i}{2m}$ when $G$ is a random graph. Let $\mathcal{H}$ denotes a binary random variable where $\mathcal{H}=1$ if the randomly chosen edge has at least one end inside community $S$ and $\mathcal{H}=0$ otherwise. Then, the partial modularity $Q_S$ can be rewritten as: $P(\mathcal{H}\mathcal{H})-P(\mathcal{H})P(\mathcal{H})$. We can find that the partial modularity owns the same idea with our measurement function in (5), i.e., they are both the concrete examples of Piatetsky-Shaprio Rule-Interest. However, the derivation of the partial modularity and our function starts from two different perspectives respectively. The partial modularity is edge-centric, which utilizes the relationship between edge and community. Our formulation in (5) is node-centric, which  utilizes the relationship between node and community. For the sake of simplicity, we will use PS to denote our measurement function in (5) and we can also call it ``node modularity".

Despite the PS function can be used for evaluating the strength of the correlation between a node and a community, it is still not clear if it satisfies  aforementioned three properties. A thorough analysis is essential to confirm the utility of PS in the context of community detection.
\textbf{P2} and \textbf{P3} are necessary properties that a good community correlation measure should have. Essentially, these two properties can be interpreted as "a  higher value of $PS(u,S)$ indicates a strong correlation between $u$ and $S$". \textbf{P1} is concerned how far the correlation between $u$ and $S$ is away from what it is in the random graph. For the PS measure, we have the following two theorems, where Theorem 1 proves that P2 and P3 hold and  Theorem 2 proves that P1 is correct.
\begin{theorem}
For fixed $N>0$, we have $0<\epsilon<N$, $0<d_u<N$ and $0<\omega\leq min(\epsilon,d_u)$. Let $PS(u,S)$ be the correlation value between $u$ and $S$ defined in (5), then 1) $PS(u,S)$ monotonically increases with the increment of $\omega$ when  $\epsilon$ and $d_u$  are fixed, 2) $PS(u,S)$ monotonically decreases with the increment of $\epsilon$ when  $\omega$ and $d_u$ are fixed, and 3) $PS(u,S)$ monotonically decreases with the increment of $d_u$ when  $\epsilon$ and $\omega$  are fixed.
\end{theorem}
\begin{proof}
Let $PS(u,S) = J(\omega,\epsilon,d_u)$. We will directly calculate the partial derivatives of $PS(u,S)$ regarding to $\epsilon$, $\omega$ and $d_u$.
\begin{equation*}
    \begin{split}
	&\frac{\partial J(\omega,\epsilon,d_u)}{\partial \omega} = \frac{1}{N}>0, \frac{\partial J(\omega,\epsilon,d_u)}{\partial \epsilon} = \frac{-d_u}{N^2}<0, \\
	& \frac{\partial J(\omega,\epsilon,d_u)}{\partial d_u} = \frac{-\epsilon}{N^2}<0.
    \end{split}
\end{equation*}
Overall, the monotonicity of $PS(u,S)$ is consistent with the monotonicity of $\omega$, $\epsilon$ and $d_u$ respectively.
\end{proof}
\begin{theorem}
Given the E-R model $G(n,m)$ with the probability $p=2m/(n(n-1))$ and let $PS(u,S)$ be the correlation value between $u$ and $S$, then $\mathbb{E}[PS(u,S)] = 0$ holds for any subgraph $S \in \mathcal{G}_S$ and node $u \in S$, where $\mathcal{G}_S$ is the set of all the subgraphs of $G$.
\end{theorem}
\begin{proof}
Let $X_v$ be a random variable such that:
$$
X_v =
\begin{cases}
1  & \text{if $v$ has a link with $u$}, \\
0  & \text{otherwise}.
\end{cases}
$$
Clearly, $\mathbb{E}[X_v] = 1\cdot \mathbb{P}(X_v=1)+0\cdot \mathbb{P}(X_v=0) = p$, $d_u=\sum_{i=1}^{n}X_i$ and $\omega = \sum_{v \in S}X_v$. By the linearity of expectations,
\begin{equation*}
    \begin{split}
        \mathbb{E}[PS(u,S)] &= \mathbb{E}[\frac{\omega}{N}-\frac{\epsilon d_u}{N^2}] = \mathbb{E}[\frac{\omega}{N}]-\mathbb{E}[\frac{\epsilon d_u}{N^2}]\\
			    &= \frac{1}{N}\mathbb{E}[\sum_{v \in S}X_v]-\frac{\epsilon}{N^2}\mathbb{E}[\sum_{i=1}^{n}X_i]\\
			    &= \frac{1}{N}\sum_{v \in S}\mathbb{E}[X_v]-\frac{\epsilon}{N^2}\sum_{i=1}^{n}\mathbb{E}[X_i]\\
			    &= \frac{\epsilon p}{N}-\frac{\epsilon}{N^2} \cdot Np = 0.
    \end{split}
\end{equation*}
\end{proof}
Theorem 2 shows that PS can be regarded as the distance between the correlation value of $u$ and $S$ in a real network and the average correlation value of $u$ and $S$ over all the random networks. However, this theoretic result ignores the variance of PS variable over all the random networks. If the distribution of PS values is not strongly peaked, it is very likely that most PS values will far exceed zero even in the random networks. Thus, we should investigate the variance of PS values as well. Then, we have the following Lemma.
\begin{lemma}
Given the E-R model $G(n,m)$ with the probability $p=2m/(n(n-1))$.  Let  $\lambda =Np$, which is the expected degree under the E-R model. $PS(u,S)$ is correlation value between $u$ and $S$. Let $\mathcal{G}_S$ be the set of all the subgraphs of $G$. Then, $\forall \kappa>0$, $\forall S \in \mathcal{G}_S$ and $\forall u \in S$, we have the upper bound on the probability that PS(u, S) is no less than $\sqrt{\kappa}$:
\begin{equation*}
    \begin{split}
         &\mathbb{P} \Big[ \big| PS(u,S) \big| \geq \sqrt{\kappa} \Big] \\ &\leq \frac{1}{\kappa} \Big( \frac{ p^2\epsilon^2-p^2\epsilon+ p\epsilon }{N^2}+
							\frac{ \lambda\epsilon^2(\lambda-p+1) }{N^4} - \frac{ 2p\epsilon^2(\lambda+1) }{N^3} \Big).
    \end{split}
\end{equation*}
\end{lemma}

\begin{proof}
 Analogously to what we have done in the proof of Theorem 2, we will use the sum of random variable $X_{i}$ to denote $\omega$ and $d_u$:
$$d_u=\sum_{i=1}^{n}X_i,\omega = \sum_{v \in S}X_v.$$
Let $Y = PS(u,S)$, the second moment of $Y$ can be obtained by the linearity of expectations:
\begin{equation*}
    \begin{split}
        \mathbb{E}[Y^2] &= \mathbb{E}[ \Big( \frac{\omega}{N}-\frac{\epsilon d_u}{N^2} \Big)^2 ] = \mathbb{E}[\frac{\omega^2}{N^2}-\frac{2\epsilon d_u\omega}{N^3}+\frac{\epsilon^2 d_u^2}{N^4}]\\
		    &= \frac{1}{N^2}\mathbb{E}[\omega^2]-\frac{2\epsilon}{N^3}\mathbb{E}[d_u \omega]+\frac{\epsilon^2}{N^4}\mathbb{E}[{d_u}^2].
    \end{split}
\end{equation*}
Calculating $\mathbb{E}[\omega^2]$, $\mathbb{E}[d_u \omega]$ and $\mathbb{E}[{d_u}^2]$ independently, we have
\begin{equation*}
    \begin{split}
        \mathbb{E}[\omega^2] &= \mathbb{E}[ \Big( \sum_{i \in S}X_i \Big) \Big( \sum_{j \in S}X_j \Big) ] = \mathbb{E}[ \sum_{i \in S} \sum_{j \in S}X_iX_j]\\
			    &= \sum_{i \in S} \sum_{j \in S} \mathbb{E}[X_iX_j] = \epsilon(\epsilon-1)p^2+p\epsilon, \\
	   \mathbb{E}[d_u \omega] &= \mathbb{E}[ \Big( \sum_{i=1}^{n}X_i \Big) \Big( \sum_{j \in S}X_j \Big) ] = \mathbb{E}[\sum_{i=1}^{n} \sum_{j \in S} X_iX_j]\\
			    &= \sum_{i=1}^{n} \sum_{j \in S} \mathbb{E}[X_iX_j] =\epsilon Np^2+p\epsilon = \epsilon \lambda p + p\epsilon,\\
	   \mathbb{E}[{d_u}^2]  &= \mathbb{E}[ \Big( \sum_{i=1}^{n}X_i \Big) \Big( \sum_{j=1}^{n}X_j \Big) ] = \sum_{i=1}^{n}\sum_{j=1}^{n} \mathbb{E}[X_iX_j]\\
			    &= N(N-1)p^2 + Np = \lambda(\lambda-p+1).
    \end{split}
\end{equation*}
To obtain the variance of $Y$, we first apply  Theorem 2:
\begin{equation*}
    \begin{split}
     Var[Y] &=  \mathbb{E}[Y^2] - (\mathbb{E}[Y])^2 =  \mathbb{E}[Y^2]\\
	&= \frac{ \epsilon(\epsilon-1)p^2+p\epsilon }{N^2}-\frac{ 2p\epsilon^2(\lambda+1) }{N^3} + \frac{ \lambda\epsilon^2(\lambda-p+1) }{N^4}.
    \end{split}
\end{equation*}
Then, by the Chebyshev inequality, we finish the proof:
\begin{equation*}
    \begin{split}
    & \mathbb{P}\Big[ \big| PS(u,S) \big| \geq \sqrt{\kappa} \Big] = \mathbb{P}\Big[ \big|Y-\mathbb{E}[Y]\big| \geq \sqrt{\kappa} \Big] \leq \frac{ Var[Y] }{\kappa}\\
    &=\frac{1}{\kappa} \Big(  \frac{ \epsilon(\epsilon-1)p^2+p\epsilon }{N^2}-\frac{ 2p\epsilon^2(\lambda+1) }{N^3} + \frac{ \lambda\epsilon^2(\lambda-p+1) }{N^4} \Big) .
    \end{split}
\end{equation*}
\end{proof}
Most large-scale real networks appear to be sparse~\cite{wang2003complex}~\cite{albert2002statistical}~\cite{del2011all}, in which the number of edges is generally the order $n$ rather than $n^2$~\cite{albert2002statistical}. In addition, sparse graphs are particularly sensitive to random fluctuations~\cite{fortunato2016community}. Therefore, we should concentrate more on the performance of PS on the sparse graph. With respect to the sparsity under the E-R model, we have the following Theorem.
\begin{theorem}
Given the ER model $G(n,m)$ with the probability $p=2m/(n(n-1))$ and $\lambda =Np$ is the expected degree under E-R model. Assuming $\lambda$ always remains finite in the limit of infinite size~\cite{fortunato2016community}. $PS(u,S)$ is the correlation value between $u$ and $S$. Then, $\forall \kappa>0$, any subgraph $S$ and $\forall u \in S$, we have:
$$
         \lim_{N\to+\infty, p \to 0} \mathbb{P} \Big[ \big| PS(u,S) \big| \geq \kappa \Big] = 0.
$$
\end{theorem}
\begin{proof}
We first apply Lemma 1:
\begin{equation*}
    \begin{split}
    & \mathbb{P} \Big[ \big| PS(u,S) \big| \geq \kappa \Big] \\
    &\leq\frac{1}{ {\kappa}^2 } \Big(  \frac{ \epsilon(\epsilon-1)p^2+p\epsilon }{N^2}-\frac{ 2p\epsilon^2(\lambda+1) }{N^3} + \frac{ \lambda\epsilon^2(\lambda-p+1) }{N^4} \Big) \\
    &\leq\frac{1}{{\kappa}^2 } \Big(  \frac{ \epsilon(\epsilon-1)p^2+p\epsilon }{N^2}+ \frac{ \lambda\epsilon^2(\lambda-p+1) }{N^4} \Big)\\
    &\leq\frac{1}{{\kappa}^2 } \Big(  \frac{ (Np)^2+pN }{N^2}+ \frac{ \lambda (\lambda-p+1) }{N^2} \Big)\\
    & = \frac{1}{{\kappa}^2 } \Big(  \frac{ 2(\lambda)^2+2\lambda-\lambda p }{N^2} \Big).
    \end{split}
\end{equation*}
Then, we take the limit:
\begin{equation*}
    \begin{split}
      &\lim_{N\to+\infty, p \to 0} \mathbb{P} \Big[ \big| PS(u,S) \big| \geq \kappa \Big] \\
      &\leq  \lim_{N\to+\infty, p \to 0} \frac{1}{{\kappa}^2 } \Big(  \frac{ 2(\lambda)^2+2\lambda-\lambda p }{N^2} \Big)= 0.
    \end{split}
\end{equation*}
\end{proof}
Theorem 3 shows that it is difficult to obtain a high PS value in the large sparse random network. This theoretic result indicates that the distribution of PS values is strongly peaked, that is, most PS values are near by zero. In fact,  even in the small random network, the upper bound introduced in Lemma 1 is often a small value. For example, let $\kappa = 0.0001$, $N=280$, $\epsilon = 20$ and $\lambda = 8$, then we have $Pr \big[ | PS(u,S) | \geq 0.01 \big] \leq 6.54\%$. Overall, the high PS value between $u$ and $S$ in the real sparse network is an indicator that there is a significant correlation between $u$ and $S$. This provides us a theoretical basis for our community detection algorithm.

To formulate a vertex-centric metric of a community, we have:
$$ F(S) = \sum_{u\in S}PS(u,S) =\frac{2l_S}{N}-\frac{\epsilon K_S}{N^2},        \eqno{(6)}$$
where  $l_S$ is the number of the edges inside community $S$ and  $K_S$ represents the sum of degrees of all vertexes in community $S$. Then, the objective function for a community partition is:
$$ \Gamma (P) = \sum_{j=1}^{\mathcal{M}} F(S_{j}) = \sum_{j=1}^{\mathcal{M}} \Big(\frac{2l_{S_j}}{N}- \frac{\epsilon K_{S_j}}{N^2} \Big).        \eqno{(7)}$$

\subsubsection{$\phi$-Coefficient}
$\phi$-Coefficient~\cite{Agresti2003Categorical} is a variant of Pearson's Product-moment Correlation Coefficient for binary variables. In association mining, it is often used to estimate whether there is a non-random pattern. Our measurement function based on $\phi$-Coefficient is calculated as:
\begin{align*}
   \phi(u,S) =& \frac{P(\mathcal{C}\mathcal{G}) - P(\mathcal{C})P(\mathcal{G}) }{ \sqrt{P(\mathcal{C})(1-P(\mathcal{C}))P(\mathcal{G})(1-P(\mathcal{G}))}}  \\
  =&\frac{\omega N - \epsilon d_u }{ \sqrt{\epsilon(N-\epsilon) d_u(N-d_u)} }. \tag{$8$}
\end{align*}
A positive $\phi(u,S)$ value indicates that node $u$ has denser intra-connection with community $S$ and stronger isolation from the rest of graph. Moreover, $\phi(u,S)$ has the same range of values as Pearson's Product-moment Correlation Coefficient, i.e. $-1\leq \phi(u,S) \leq1$. When $\phi(u,S) = 1$, node $u$ will have all its links connecting all the members of community $S$ and have no other links with the rest of graph. The above discussions are subjective and intuitive. To further analyze $\phi(u,S)$, we have some similar theoretic results as PS owns.
\begin{theorem}
For fixed $N>0$, we have $0<\epsilon<N$, $0<d_u<N$ and $0<\omega\leq min(\epsilon,d_u)$. Let $\phi(u,S)$ be $\phi$-Coefficient value between $u$ and $S$ defined in (8), then 1) $\phi(u,S)$ monotonically increases with the increment of $\omega$ when  $\epsilon$ and $d_u$  are fixed, 2) $\phi(u,S)$ monotonically decreases with the increment of $\epsilon$ when  $\omega$ and $d_u$ are fixed, and 3) $\phi(u,S)$ monotonically decreases with the increase of $d_u$ when  $\epsilon$ and $\omega$  are fixed.
\end{theorem}
\begin{proof}
Let $\phi(u,S) = J(\omega,\epsilon,d_u)$. We will directly calculate the partial derivatives of $\phi(u,S)$ regarding to $\epsilon$, $\omega$ and $d_u$.
\begin{equation*}
    \begin{split}
	&\frac{\partial J(\omega,\epsilon,d_u)}{\partial \omega} = \frac{N}{\sqrt{\epsilon(N-\epsilon) d_u(N-d_u)}}>0\\
	& \frac{\partial J(\omega,\epsilon,d_u)}{\partial \epsilon} = {\big( d_u(N-d_u) \big) }^{ - \frac{1}{2}} \cdot \frac{-N\big[\epsilon(d_u-2\omega)+\omega N \big]}{2{\big( \epsilon(N-\epsilon) \big) }^{ \frac{3}{2}}} \\
	& \frac{\partial J(\omega,\epsilon,d_u)}{\partial d_u} = {\big( \epsilon(N-\epsilon) \big) }^{ - \frac{1}{2}} \cdot \frac{-N\big[d_u(\epsilon-2\omega)+\omega N \big]}{2{\big( d_u(N-d_u) \big) }^{ \frac{3}{2}}}.
    \end{split}
\end{equation*}
We only need to focus on the term $-N\big[d_u(\epsilon-2\omega)+\omega N \big]$. This term is obviously negative  when $\epsilon \geq 2\omega$. Now assuming $\epsilon<2\omega$. Since $\omega \leq \epsilon$, we have $2\omega - \epsilon \leq \omega$. Then,
\begin{equation*}
    \begin{split}
	-N\big[d_u(\epsilon-2\omega)+\omega N \big] &= -N\big[\omega N - d_u(2\omega-\epsilon)\big]\\
									 &\leq -N\big(\omega N - d_u\omega\big) \\
 									 &< -N\big(\omega N - N\omega\big) = 0.
    \end{split}
\end{equation*}
Thus, $\frac{\partial J(\omega,\epsilon,d_u)}{\partial \epsilon}<0$ holds. In a similar way, $\frac{\partial J(\omega,\epsilon,d_u)}{\partial d_u}<0$ holds as well. Overall, the monotonicity of $\phi(u,S)$ is consistent with the monotonicity of $\omega$, $\epsilon$ and $d_u$ respectively.
\end{proof}

\begin{theorem}
Given the E-R model $G(n,m)$ with the probability $p=2m/(n(n-1))$ and let $\phi(u,S)$ be $\phi$-Coefficient value between $u$ and $S$. For fixed $0<\epsilon<N$, $\mathbb{E}[\phi(u,S)] = 0$ holds for any subgraph $S \in \mathcal{G}_S$ and node $u \in S$, where $\mathcal{G}_S$ is the set of all the subgraphs of $G$.
\end{theorem}
\begin{proof}
Firstly, we introduce a random variable:
$$Y=\frac{1}{\sqrt{d_u(N-d_u)}}.$$
At the same time, let $X_v$ be a random variable such that:
$$
X_v =
\begin{cases}
1 & \text{if $v$ has a link with $u$}, \\
0  & \text{otherwise}.
\end{cases}
$$
Integrating these two random variables, we have:
\begin{equation*}
    \begin{split}
		\mathbb{E}[ X_vY ] &= \mathbb{E}[ Y|X_v=1 ] \cdot \mathbb{P}(X_v=1)+ 0 \cdot \mathbb{P}(X_v=0)\\
			 &= p\cdot \mathbb{E}\Big[\frac{1}{\sqrt{d_u(N-d_u)}}|X_v=1 \Big] = \mathbb{E}[H_v]p,
    \end{split}
\end{equation*}
where $\mathbb{E}[H_v] = \mathbb{E}[Y|X_v=1]$. Besides, we have $d_u=\sum_{i=1}^{n}X_i$ and $\omega = \sum_{v \in S}X_v$. Let $\mathcal{K} = {\big( \epsilon(N-\epsilon) \big) }^{ - \frac{1}{2}}$.  By the linearity of expectations,
\begin{equation*}
    \begin{split}
        \mathbb{E}[\phi(u,S)] &= \mathbb{E}[\frac{\omega N - \epsilon d_u }{ \sqrt{\epsilon(N-\epsilon) d_u(N-d_u)} }]\\
			    &= \mathcal{K}  \cdot \mathbb{E}[(\omega N - \epsilon d_u )Y]\\
			    &= \mathcal{K}  \cdot \mathbb{E}[(N\sum_{v \in S}X_v - \epsilon \sum_{i=1}^{n}X_i ){ Y }]\\
			    &= \mathcal{K}  \cdot \Big( N \cdot \mathbb{E}[Y{\sum_{v \in S}X_v}] - \epsilon  \cdot \mathbb{E}[Y{ \sum_{i=1}^{n}X_i }] \Big) \\
			    &= \mathcal{K}  \cdot \mathbb{E}[H_V]\cdot \Big( N\epsilon p  -\epsilon Np \Big) = 0 .
    \end{split}
\end{equation*}
\end{proof}
Theorem 4 and Theorem 5 indicate that $\phi(u,S)$ is a good community measurement function since it owns the necessary properties \textbf{P1},\textbf{P2} and \textbf{P3}. This is  the theoretic basis for $\phi(u,S)$ to be employed in community detection. Now we can take the sum of $\phi(u,S)$ over all $u \in S$ to formulate a vertex-centric metric:
\begin{align*}
   \Phi(S) &= \sum_{u\in S}\phi(u,S) = \sum_{u\in S}\frac{\omega N - \epsilon d_u }{ \sqrt{\epsilon(N-\epsilon) d_u(N-d_u)} } \\
     		&=  {\big( \epsilon(N-\epsilon) \big) }^{ - \frac{1}{2}}\sum_{u\in S}\frac{\omega N - \epsilon d_u }{ \sqrt{d_u(N-d_u)} }. \tag{$9$}
\end{align*}

Despite different definitions on what a community should be have been proposed, some general ideas are widely accepted by most scholars. One of them is that the clique can be regarded as a perfect community. Thus, the loosely connected cliques should be separated from each other as different communities. The modularity function may fail to detect modules which are smaller than a scale in large networks, which is called the resolution limit problem~\cite{Fortunato2007Resolution}~\cite{Lancichinetti2011Limits}. A more general example in real network has been given in~\cite{Fortunato2007Resolution}, which is shown in Fig. 5. There are two cliques of the same size $\mathcal{C}_s$ and $\mathcal{C}_t$, and there is an edge connecting one node $u$ from $\mathcal{C}_s$ and another node $v$ from $\mathcal{C}_t$.
\begin{figure}[!ht]
    \centering
       \begin{overpic}[scale=0.5]{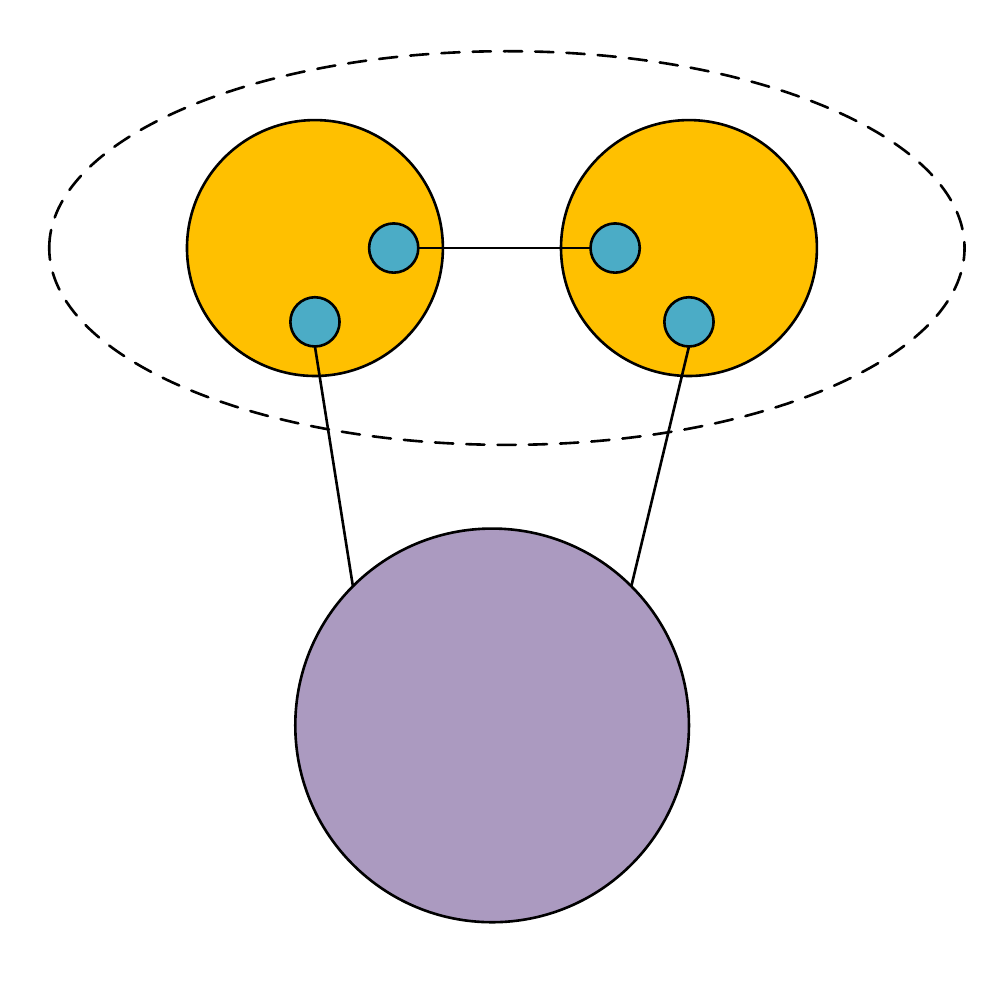}
			\put(10,70){\scalebox{1.3}{$\mathcal{C}_s$}}
			\put(85,70){\scalebox{1.3}{$\mathcal{C}_t$}}
			\put(36,78){\scalebox{1.3}{$u$}}
			\put(60,78){\scalebox{1.3}{$v$}}
			\put(69,70){\scalebox{1.3}{${o}$}}
			\put(30,70){\scalebox{1.3}{${c}$}}
		\end{overpic}
	  \caption{ A network with two equal-sized cliques $\mathcal{C}_s$ and $\mathcal{C}_t$, in which there is only one edge between these two cliques. In addition, there is only one edge between each clique and the rest of graph. If network is large enough, modularity optimization will merge these two small modules into one community.}
	  \label{fig5}
\end{figure}
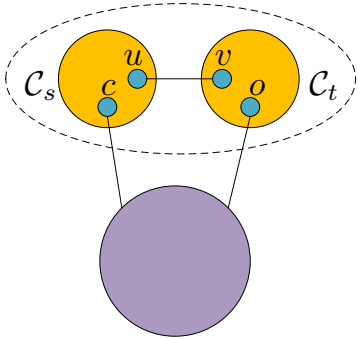
When the network is large enough, the modularity-based algorithms prefer to merge these two cliques to form a bigger community since this operation will help increase the value of modularity function. Apparently, it is counter-intuitive and disturbing. Fortunately, we can prove that such issue does not occur in $\Phi(S)$.
\begin{theorem}
Given a graph $G(n,m)$ with a subgraph which consists of two equal-sized cliques $\mathcal{C}_s$ and $\mathcal{C}_t$. There is only one edge connecting $u$ and $v$, where $u\in \mathcal{C}_s $ and $v\in \mathcal{C}_t $. Let $P = \mathcal{C}_s \bigcup \mathcal{C}_t$ and $\mathcal{I}=|\mathcal{C}_s|=|\mathcal{C}_t|$. Two nodes $c$ and $o$, where $c\in \mathcal{C}_s$ and $o\in \mathcal{C}_t$, both have one edge connecting the rest of graph respectively. Assuming $5\leq \mathcal{I}<\frac{n}{4}$, then we have $\Phi(\mathcal{C}_s)+\Phi(\mathcal{C}_t)>\Phi(P)$.
\end{theorem}
\begin{proof}
From $\mathcal{I}<\frac{n}{4}$, we can know that $\epsilon<\frac{n}{4}$. Let $\Delta\Phi = \Phi(P) - ( \Phi(\mathcal{C}_s)+\Phi(\mathcal{C}_t))$ and $T = 2\epsilon+1$, then $\Delta\Phi$ can be calculated as:
\begin{equation*}
    \begin{split}
	    \Delta\Phi &= {\big( T(N-T) \big) }^{ - \frac{1}{2}}\sum_{i\in P} \phi(i,P) - 2\Phi(\mathcal{C}_s)\\
			    &=  2\sum_{i\in \mathcal{C}_s} \frac{\omega_i N - \epsilon d_i }{ \sqrt{d_i(N-d_i)} } \Big( {\big( T(N-T) \big) }^{ - \frac{1}{2}}-{\big( \epsilon(N-\epsilon) \big) }^{ - \frac{1}{2}}  \Big)\\
&+2\frac{ N }{ \sqrt{d_u(N-d_u)} }{\big( T(N-T) \big) }^{ - \frac{1}{2}}.
\end{split}
\end{equation*}
Let $\mathcal{L} = \frac{ N }{ \sqrt{d_u(N-d_u)} }{\big( T(N-T) \big) }^{ - \frac{1}{2}}$ and $\mathcal{B} = \sum_{i\in \mathcal{C}_s}$ $ \frac{\omega_i N - \epsilon d_i }{ \sqrt{d_i(N-d_i)} }$, we have:
\begin{equation*}
    \begin{split}
	    \frac{1}{2}\Delta\Phi &= \mathcal{B}( \frac{ \sqrt{\epsilon(N-\epsilon)} - \sqrt{T(N-T)} }{\sqrt{T(N-T)\epsilon(N-\epsilon)} } )+\mathcal{L}.
\end{split}
\end{equation*}
Since $T=2\epsilon+1 \leq \frac{n}{2}$, then $\sqrt{2\epsilon(N-2\epsilon)}<\sqrt{T(N-T)}$. We can obtain:
\begin{equation*}
    \begin{split}
	    \frac{1}{2}\Delta\Phi &< -\mathcal{B}( \frac{  \sqrt{2\epsilon(N-2\epsilon)}-\sqrt{\epsilon(N-\epsilon)}}{\sqrt{T(N-T)\epsilon(N-\epsilon)} } )+\mathcal{L}.
\end{split}
\end{equation*}
Now we need to analyze $\mathcal{B}$. Firstly, $\mathcal{B}$ can be calculated as:
\begin{align*}
  \mathcal{B} &= \sum_{i\in \mathcal{C}_s} \frac{\omega_i N - \epsilon d_i}{\sqrt{ d_i(N-d_i) }} \xlongequal{\omega_i=\epsilon}
       \sum_{i\in \mathcal{C}_s} \frac{\epsilon(N -  d_i)}{\sqrt{ d_i(N-d_i) }}\\
    &= \sum_{i\in \mathcal{C}_s} \epsilon \sqrt{\frac{ N-d_i}{d_i}}.
\end{align*}
We have $d_i = \omega_i = \epsilon$ for $i \neq u$ and $i \neq c$ and $d_c=d_u = \epsilon+1$, then:
\begin{align*}
\mathcal{B} =& (\epsilon-1) \sqrt{\epsilon(N-\epsilon)}+ 2\epsilon \sqrt{\frac{N-(\epsilon+1)}{\epsilon+1}}\\
		  >& (\epsilon-1)\sqrt{\epsilon(N-\epsilon)}.
\end{align*}
Next, we will analyze $\mathcal{L}$ as well. Since $\epsilon<d_u<\frac{n}{2}$, we have:
$$\mathcal{L}<\frac{ N }{ \sqrt{\epsilon(N-\epsilon)} }{\big( T(N-T) \big) }^{ - \frac{1}{2}}.$$
Thus, we have a upper bound of $\frac{1}{2}\Delta\Phi$ as:
\begin{equation*}
    \begin{split}
	    \frac{1}{2}\Delta\Phi <& -(\epsilon-1) \sqrt{\epsilon(N-\epsilon)}( \frac{\sqrt{2\epsilon(N-2\epsilon)}
               -\sqrt{\epsilon(N-\epsilon)}}{\sqrt{T(N-T)\epsilon(N-\epsilon)} } )\\
               &+ \frac{N}{\sqrt{\epsilon(N-\epsilon)T(N-T)}} \\
             =&  \frac{ -(\epsilon-1) \sqrt{\epsilon(N-\epsilon)} \big( \sqrt{2\epsilon(N-2\epsilon)}-\sqrt{\epsilon(N-\epsilon)} \big) +N}{\sqrt{\epsilon(N-\epsilon)T(N-T)}}\\
		  =& \frac{ -\epsilon(\epsilon-1)(N-\epsilon)\big(\sqrt{2(2-\frac{N}{N-\epsilon})}-1 \big) +N}{\sqrt{\epsilon(N-\epsilon)T(N-T)}}\\
		  <& \frac{ -\epsilon(\epsilon-1)(N-\epsilon)\big( \sqrt{\frac{4}{3}}-1 \big) +N}{\sqrt{\epsilon(N-\epsilon)T(N-T)}}.
\end{split}
\end{equation*}
Let $g(\epsilon) = -\epsilon(\epsilon-1)(N-\epsilon)\big( \sqrt{\frac{4}{3}}-1 \big) +N$. Since $4 \leq \epsilon<\frac{n}{2}$, $g(\epsilon)$ decreases as $\epsilon$ increases. Then, we have:
$$ g(\epsilon) \leq g(4) < -1.8(N-4)+N<0.$$
Thus, $\Delta\Phi<0$. We finish the proof.
\end{proof}
Theorem 6 reveals the fact that $\Phi(S)$ can avoid merging  two very pronounced communities.  The reason is that $\Phi(S)$ takes into account the local structural information of every vertex. Now we will provide a comprehensible interpretation on this point. Firstly, it can be found that all the $\phi(u,S)$ values of vertices from $\mathcal{C}_s$ are 1 or approximately 1. The value 1 is the maximum value that a vertex can achieve, which indicates that a node completely belongs to a community in the sense that this node connects every member of the community and has no external links. Obviously, the $\phi(u,S)$ value will be reduced after merging $\mathcal{C}_s$ and $\mathcal{C}_t$ if the number of edges across $\mathcal{C}_s$ and $\mathcal{C}_t$ are too few. If we want to retain the correlation value 1 for each vertex, the vertices from $\mathcal{C}_s$ should connect all vertices from $\mathcal{C}_t$. However, it is a far cry from the situation shown in Fig. 6.  We can observe that the number of inter-edges across two communities which allows two communities to merge is affected by the $\Phi(S)$ values in two communities. When $\Phi(S)$ values of two communities are both higher, it will require more inter-edges.

One of the main causes of the resolution limit problem is that the modularity only considers the whole community rather than every vertex. Besides, it only depends on the number of edges $m$ when the network is sufficiently large. Let $Q$ be the modularity sum of two communities before the mergence and $\acute{Q}$ be the modularity value after the union. Then, we have:
$$\Delta Q =\acute{Q}-Q = \frac{\Delta l}{2m}-\frac{K_1K_2}{2m^2},$$
where $K_1$ is the degree sum of community 1, $K_2$ is the degree sum of community 2 and $\Delta l$ is the difference between  the number of intra-edges after merging two communities and that before the combination. If $\Delta Q>0$, it requires $\Delta l>\frac{K_1K_2}{2m}$. As we can see, when $m$ increases, the number of inter-edges needed will be reduced. As $m$ tends to infinity, any two communities will be merged even if there is only one edge  between them. The reason why this happens is that $\Delta l$ only embodies the overall structural information of a community and the modularity misses local structural information. To avoid the resolution limit problem as much as possible, first of all, the community metric should consider every vertex rather than view the community as a single unit. Second, the community metric should take full advantage of the local structural information. Then, we have the following Theorem.

\begin{theorem}
Given a graph $G(n,m)$ with a vertex $u$ and a community $S$. To simplify the notations, let $\boldsymbol{\Psi}=\Psi_{S\backslash \{ u \}}$ and $\boldsymbol{\psi}=\psi_u$, where $\boldsymbol{\Psi}$ is the community vector of S and $\boldsymbol{\psi}$ is the neighbor vector of $u$. For fixed $\epsilon>0$, $d_u>0$ and $\omega>0$, we have:
$$\lim_{N\to+\infty}\phi(u,S) = \frac{\boldsymbol{\Psi} \cdot \boldsymbol{\psi}}{ \lVert \boldsymbol{\Psi} \rVert_{2} \lVert \boldsymbol{\psi}\rVert_{2}},$$
where $\lVert . \rVert_{2}$ is Euclidean norm.
\end{theorem}
\begin{proof}
We directly calculate the limit of $\phi(u,S)$ as $N$ tends to infinity.
\begin{equation*}
    \begin{split}
	    \lim_{N\to+\infty}\phi(u,S) =&\lim_{N\to+\infty}\frac{\omega N - \epsilon d_u }{ \sqrt{\epsilon(N-\epsilon) d_u(N-d_u)} }\\
							=&\lim_{N\to+\infty}\frac{\omega - \frac{\epsilon d_u}{N} }{ \sqrt{\epsilon(1-\frac{\epsilon}{N}) d_u(1-\frac{d_u}{N})} }\\
							=&\frac{\omega}{ \sqrt{\epsilon d_u} } = \frac{\sum_{i=1}^{N} \boldsymbol{\Psi}_i \boldsymbol{\psi}_i}{ \sqrt{\sum_{i=1}^{N} \boldsymbol{\Psi}_{i}^2} \sqrt{\sum_{i=1}^{N} \boldsymbol{\psi}_{i}^2 }}\\
							=&\frac{\boldsymbol{\Psi} \cdot \boldsymbol{\psi}}{ \lVert \boldsymbol{\Psi} \rVert_{2} \lVert \boldsymbol{\psi}\rVert_{2}}.
    \end{split}
\end{equation*}
\end{proof}
Theorem 7 indicates that $\phi(u,S)$ is the Cosine similarity between $\boldsymbol{\Psi}$ and $\boldsymbol{\psi}$  as $N$ tends to infinity. Then, we can recalculate $\Phi(S)$ in the limit:
$$\Phi(S) = Cos(S) = \sum_{u\in S}\frac{\omega_u}{ \sqrt{\epsilon d_u} }.    \eqno{(10)}$$
Obviously, it considers the Cosine similarity values between each node and  the community. As a result, it is unlikely to perform the community mergence operation  in the case of weak inter-connections between two communities. Theorem 6 and Theorem 7 both indicate $\Phi(S)$ can mitigate the resolution limit problem.

\subsection{Summary}
As a short summary, we would like to present the following remarks.  First of all, the use of different correlation functions in our node-centric framework
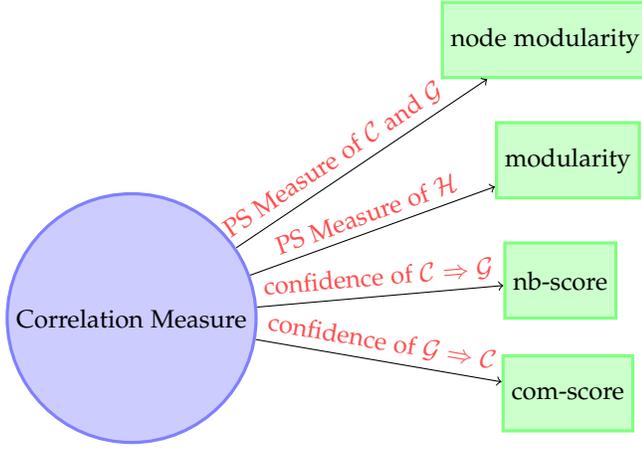
\begin{figure}[!ht]
\centering
\begin{tikzpicture}
    [L1Node/.style={circle,   draw=blue!50, fill=blue!20, very thick, minimum size=10mm},
    L2Node/.style={rectangle,draw=green!50,fill=green!20,very thick, minimum size=10mm}]
       \node[L1Node] (n1) at (0, 0){Correlation Measure};
       \node[L2Node] (n2) at (5.5, 3.7){node modularity};
	 \node[L2Node] (n3) at (5.8, 2.1){modularity};
	 \node[L2Node] (n4) at (5.7, 0.5){nb-score};
	 \node[L2Node] (n5) at (5.8, -1.0){com-score};
	\draw [->] (n1)--node [color=red!70,pos=0.43,above,sloped]{PS Measure of $\mathcal{C}$ and $\mathcal{G}$} (n2);
	\draw [->] (n1)--node [color=red!70,pos=0.5,above,sloped]{PS Measure of  $\mathcal{H}$} (n3);
	\draw [->] (n1)--node [color=red!70,pos=0.5,above,sloped]{confidence of $\mathcal{C} \Rightarrow \mathcal{G}$ } (n4);
	\draw [->] (n1)--node [color=red!70,pos=0.5,above,sloped]{confidence of $\mathcal{G}\Rightarrow \mathcal{C}$ } (n5);
    \end{tikzpicture}
 \caption{ The use of different correlation functions may lead to several different known community evaluation measures.}
	  \label{fig5}
\end{figure}
may lead to some different but known community evaluation measures.  As shown,  modularity, node modularity, neighborhood connectedness and community connectedness in Focs \cite{bandyopadhyay2015focs} are all concrete examples in our abstract framework. More importantly, we may explore more correlation functions to obtain more effective community evaluation measures in the future.

\section{Correlation-Based Community Detection}
In this section, we propose a  Correlation-Based Community Detection (CBCD) algorithm, which is based on PS measure and $\phi$-Coefficient.  CBCD takes a graph $G(V,E)$ as input and generates a partition of $G$ which is the set of detected communities. The algorithm is divided into three phases: Seed Selection, Local Optimization Iteration and Community Merging.
\subsection{Seed Selection}
The goal of Seed Selection is to initialize a set of seed communities for our next phase. Our algorithm adopts a local search strategy to optimize the objective function defined in (7), which requires seed nodes to start with. Triadic closure is an important property of social network, which describes the basic process of social network formation~\cite{granovetter1977strength}~\cite{rapoport1953spread}~\cite{kossinets2006empirical}. Triangle (clique of size 3) in the network embodies this property. A node contained by a large number of triangles has higher probability of being the core of a community. Thus, we will use the number of triangles as the criterion to select the seed nodes. The specific process is described in Algorithm 1. We first count the number of triangles of every node in the graph and then sort the nodes by their triangle numbers decreasingly. For those nodes with the same number of triangles, we compare their degrees. Then, we go through all the nodes in this order. For every node $u$ that has not been previously visited, we mark $u$ and its neighbor nodes as visited and then put $u$ into the seed node set $P$.
\begin{algorithm}[htb]
  \caption{Seed Selection.}
  \label{alg:Framwork}
  \begin{algorithmic}[1]
    \Require
      A graph $G(V,E)$.
    \Ensure
      The set of seed nodes $P$.
	
    \State Initial $P = \emptyset$.
    \State Triangle-Counting($G$).
    \State Sort $V$ by the number of triangles decreasingly.
    \For{each $u \in V$}
      \If {$not$  $visited(u)$}
		\State Mark $u$ visited.
		\For{each $v \in neighbors(u)$}
			\State Mark $v$ visited.
		\EndFor
		\State $P = P \bigcup \{u\} $.
	 \EndIf
    \EndFor

   \\
   \Return $P$;
  \end{algorithmic}
\end{algorithm}

Triangle counting is an important task in data mining and network science~\cite{al2018triangle}, which has been widely used in many applications such as community detection and link prediction. There are many exact triangle counting algorithms in the literature~\cite{itai1978finding}~\cite{schank2005finding}~\cite{alon1997finding}~\cite{latapy2008main}. Here we modify the Ayz-Node-Counting algorithm~\cite{alon1997finding} to implement the Triangle-Counting function in Algorithm 1. Ayz-Node-Counting divides the nodes into two parts with a degree threshold $\beta$: one set of nodes whose degrees are at most $\beta$ and another set of nodes whose degrees are at least $\beta$.  Ayz-Node-Counting enumerates node-pairs that are adjacent to each node from the low degree node set. For the subgraph $\mathcal{G}$ induced from the high degree node set, Ayz-Node-Counting uses the fast matrix product to compute the number of triangles for each node in $\mathcal{G}$. By choosing appropriate $\beta$,  the worst time complexity of Ayz-Node-Counting is  $O(m^{1.4})$. Since matrix multiplication-based methods may require large memory due to adjacency matrix storage, we also enumerate over node-pairs in the induced subgraph $\mathcal{G}$. Since the total degree is $2m$, then the number of high degree nodes is at most $\frac{2m}{\beta}$. The worst time complexity of Algorithm 2 is $O(\sum_{u \in \mathcal{U}}deg(u)^2) + O( (\frac{2m}{\beta})^3)$. Since we have:
$$\sum_{u \in \mathcal{U}}deg(u)^2 \leq \sum_{u \in \mathcal{U}}deg(u) \beta \leq \beta \sum_{u \in V}deg(u) = 2m\beta,$$
 the time complexity can be written as $O(2m\beta) + O( (\frac{2m}{\beta})^3)$. Let $\beta = \sqrt{2m}$, we have time complexity $O(m\sqrt{m})$. In fact, in large sparse network, the number of high degree nodes is small even for low $\beta$. In this case, the time complexity of Triangle Counting can be viewed as $O(\beta m)$.
\begin{algorithm}[htb]
  \caption{Triangle Counting.}
  \label{alg:Framwork}
  \begin{algorithmic}[1]
    \Require
      A graph $G(V,E)$.
    \Ensure
      A array $TC$ such that $TC[u]$ is the number of triangles containing $u$.
	
    \State Initial $TC[u]=0$ for all $u$, $\mathcal{U} = \emptyset$ and $\mathcal{V} = \emptyset$.
    \For{each $u \in V$}
    	\If {$deg(u)\leq\beta$}
	    \State Add $u$ to $\mathcal{U}$.
	\Else
	    \State Add $u$ to $\mathcal{V}$.
	\EndIf	
    \EndFor
	
	\For{each $u \in \mathcal{U}$}
    	    \For{each pair $(v,w)$ of neighbors of $u$}
    		  \If {\textit{$(v,w)$ exist an edge} \textbf{and} $v<w$}
	    		\If {$deg(v)<\beta$ \textbf{and} $deg(w)<\beta$ }
	   			 \If{$u<v$}
				    \State Increase $TC[u]$, $TC[v]$ and $TC[w]$ by 1.
				 \EndIf
			\ElsIf{$deg(v)<\beta$}
	   			 \If{$u<v$}
				    \State Increase $TC[u]$, $TC[v]$ and $TC[w]$ by 1.
				 \EndIf	
			\ElsIf{$deg(w)<\beta$}
	   			 \If{$u<w$}
				    \State Increase $TC[u]$, $TC[v]$ and $TC[w]$ by 1.
				 \EndIf
			\Else
				\State Increase $TC[u]$, $TC[v]$ and $TC[w]$ by 1.
			\EndIf
		  \EndIf	
    	    \EndFor
     \EndFor
	\State Induce a subgraph $\mathcal{G}(\mathcal{V},E')$ by $\mathcal{V}$.
	\For{each $u \in \mathcal{U}$}
    	    \For{each pair $(v,w)$ of neighbors of $u$ in $\mathcal{G}$}
    		  \If {\textit{$(v,w)$ exist an edge} \textbf{and} $v<w$}
	    		\If {$u<v$ }
	   			\State Increase $TC[u]$, $TC[v]$ and $TC[w]$ by 1.
			\EndIf
		  \EndIf	
    	    \EndFor
     \EndFor
   \\
   \Return $TC$;
  \end{algorithmic}
\end{algorithm}
\subsection{Local Optimization Iteration}
In this phase, we aim at finding a partition that maximizes the objective function defined in (7). Let $P$ be a partition of graph $G$ and $P(u)$ represents the index of the community of a node $u$. Since the local search strategy (Line 5 to 20) is adopted in this phase, we initially start with a set of seed nodes. Suppose we obtain a partition $P = \{ S_{1},...,S_{\mathcal{M}} \}$ with $|P| = \mathcal{M}$ from Seed Selection. Despite we call $P$ the ``\textit{partition}", every $S_j \in P$ has only one seed node and many other nodes are not assigned to any community. To deal with this initial state, for each node $u$ that is not assigned to any community, we let $P(u) = -1$. Then, we reformulate the objective function (7) to make it convenient for implementing Local Optimization Iteration:
\begin{align*}
		&\argmax_{P}\Gamma (P) \\
		&= \argmax_{P} \sum_{ u \in V}\sum_{j=1}^{\mathcal{M}} \delta \big( P(u),j \big)\Big( \frac{l_{u,S_j}}{N}- \frac{\epsilon_{S_j} d_u}{N^2} \Big),  \tag{$10$}
\end{align*}
where $l_{u,S_j}$ is the number of edges between $u$ and $S_j$, $\delta(x,y)$ is the kronecker delta function whose value is 1 if $x=y$ and 0 otherwise.
Our optimization algorithm is an iterative process. In each iteration, it will reorganize the partition to improve the value of (9) until a locally optimal solution is achieved. First, we should consider how to assign those nodes that are not contained in any community. According to (9), for any partition $P$, we have:
\begin{align*}
      \Gamma (P) &= \sum_{ u \in V}\sum_{j=1}^{\mathcal{M}} \delta \big( P(u),j \big)\Big( \frac{l_{u,S_j}}{N}- \frac{\epsilon_{S_j} d_u}{N^2} \Big)  \\
	   &\leq \sum_{ u \in V}\max_{S \in AC(u) }\Big( \frac{l_{u,S}}{N}- \frac{\epsilon_{S} d_u}{N^2} \Big), \tag{$11$}
\end{align*}
where $AC(u)$ is the set of communities whose nodes are adjacent to node $u$. Obviously, (10) is more easily to find the maximum value than (9) if we do not modify the current partition and the PS value of each node is computed independently from others. Thus, we will find a community $S$ for each unassigned node $u$ by maximizing $PS(u,S)$ (Line 7 to 13). Then, we put node $u$ into community $S$ and update $\Gamma(P)$ with the difference brought by this operation (Line 14 to 18):
\begin{align*}
   \Delta F &= F(S') - F(S) = \frac{2l_{S'}}{N}-\frac{\epsilon_{S'} K_{S'}}{N^2} - \big(  \frac{2l_S}{N}-\frac{\epsilon_{S} K_S}{N^2}  \big)\\
            &= \frac{2(l_{S'}-l_S)}{N} - \frac{\epsilon_{S'} K_{S'}-\epsilon_{S} K_S}{N^2}\\
		 &= \frac{2l_{u,S}}{N}- \frac{ (\epsilon_S+1)(K_S+d_u)-\epsilon_S K_S}{N^2}\\
		 &=\frac{2l_{u,S}}{N}-\frac{(\epsilon_S+1)d_u+K_S}{N^2}.
\end{align*}
In practice, for each node, the partition is modified after performing the steps described in Line 7 to 18, which is to make the algorithm more robust to local maxima. Next, we consider the nodes that have already been assigned to a community. Line 21 to 39 in Algorithm 3 describes the partition refinement step. It refines the partition obtained in local search step (Line 5 to 20) using a hill climbing method. In each iteration, we perform the movements of nodes between communities to improve the value of $\Gamma(P)$. For each node $u$ belonging to a community $S$, we compute the difference brought by removing $u$ from $S$, $\Delta F = F(S) - F(S')$ (Line 23 to Line 25). For a community $S_j$, we compute the difference brought by adding $u$ to $S_j$,  $\Delta F_j = F(S_j') - F(S_j)$. We choose a community $S_i$ such that $\Delta F_i > \Delta F$ and $\Delta F_i - \Delta F$ are maximized. Then,  we add $u$ to $S_i$ and remove $u$ from $S$ to update the partition. The local search step and the partition refinement step are performed alternately until every node has been assigned to a community and objective function $\Gamma(P)$ converges.
\begin{algorithm}[htb]
  \caption{Local Optimization Iteration.}
  \label{alg:Framwork}
  \begin{algorithmic}[1]
    \Require
      A graph $G(V,E)$ and a seed set $P$.
    \Ensure
      A partition of graph $G$.
    \State Initial $P_0 = P$, $k = 0$ and $\Gamma(P_k) = 0$.
    \Repeat
	
	\State $k++$.  	
	\State $P_{k} = P_{k-1}$.
       \For{each $u \in V$}
 	 	 \If{ \textit{$u$ not in any community}}
			\State $i=-1$, $t = -\infty$.
     	 		\For{each $S \in neighbor\_community(u)$}
	 			\State Compute $ PS =\frac{l_{u,S}}{N}-\frac{\epsilon_S d_u}{N^2}$.
				\If{$PS>t$}
					\State$i=index(S)$, $t = PS$.
				\EndIf
			\EndFor
			\If{$i \ne -1$}
				\State $\Gamma(P_k)= \Gamma(P_k)+\frac{2l_{u,S}}{N}-\frac{(\epsilon_S+1)d_u+K_S}{N^2}$.
				\State Add $u$ to $S_i$.
				\State Update $P_k$.
			\EndIf
		\EndIf
	  \EndFor
	  \For{each $u \in V$}
 	 	 \If{ \textit{$u$ assigned to a community} }
			\State $i = -1$, $j =P_k(u)$.
			\State $t = 0$, $H = S_j\backslash \{u\}$.
			\State $M = \frac{2l_{u,H}}{N}-\frac{\epsilon_{H}d_u+K_{H}}{N^2}$.
     	 		\For{each $S \in neighbor\_community(u)$}
	 			\State Compute $\Delta \Gamma = \frac{2l_{u,S}}{N}-\frac{(\epsilon_S+1)d_u+K_S}{N^2}-M$.
				\If{$\Delta \Gamma>t$}
					\State$i=index(S)$, $t = \Delta \Gamma$.
				\EndIf
			\EndFor
			\If{$i \ne -1$}
				\State $\Gamma(P_k)= \Gamma(P_k)+t$.
				\State Add $u$ to $S_i$.
				\State Remove $u$ from $S_j$.
				\State Update $P_k$.
			\EndIf
		\EndIf
	  \EndFor
	\Until{$\Gamma(P_k)<\Gamma(P_{k-1})$}

   \\
   \Return $P_{k-1}$;
  \end{algorithmic}
\end{algorithm}
\subsection{Community Merging}
\begin{algorithm}[htb]
  \caption{Community Merging.}
  \label{alg:Framwork}
  \begin{algorithmic}[1]
    \Require
      A graph $G(V,E)$ , a threshold $Th$ and a partition $P$.
    \Ensure
      Detected communities.
    \State Initialize a max-heap $Max\_H = \emptyset$.
    \State Construct a community graph $\mathcal{F}(P,\Delta W)$ by $G$ and $P$.
    \For{each $S_u \in P$}
     	 \For{each $S_v \in neighbors(S_u)$}
	     \If{$\Delta W_{uv} > Th$}
			\State $Max\_H$. push( $\langle u,v,W_{uv}\rangle$ ).
		\EndIf
	 \EndFor
    \EndFor
    \While{ \textit{$Max\_H$ not empty} }
          \State $\langle u,v,{\Delta W_{uv}'}\rangle$  = $Max\_H$.top( ).
	     \State $Max\_H$.pop( ).
		\If{$S_u$ \textit{exist} \textbf{and} $S_v$ \textit{exist} \textbf{and} $\Delta W_{uv}' = \Delta W_{uv}$}
				\State Union $S_v$ and $S_u$.
				\For{each $S_w \in neighbors(S_v)$}
					\State Update $\Delta W_{uw}$ and $\Delta W_{wu}$.
					\If{$\Delta W_{uw} > Th$}
						\State $Max\_H$.push( $\langle u,w,W_{uw}\rangle$ ).
					\EndIf
				\EndFor
		\EndIf
    \EndWhile

				\State Update $P$.
   \\
   \Return $P$;
  \end{algorithmic}
\end{algorithm}
After Local Optimization Iteration, there may be many small but significant communities. We need a merging operation to find communities with suitable size. The theoretical basis of Community Merging comes from section 3.2.2. $\Phi(S)$ is the criterion used in Community Merging to judge whether two communities should be merged. Note that here $\Phi(S)$ is just used to implement the merging operation instead of being an objective function. In the whole process, we maintain a max-heap which contains a set of elements and supports delete or insert operation in $O(\log n)$ time. Community Merging is described in Algorithm 4. First, we start off with each community $S_i \in P$ being the sole node of a graph $\mathcal{F}$ and establish  an edge between $S_i$ and $S_j$ if at least one edge links them. $\mathcal{F}$ is a weighted graph, in which the weight between community $i$ and $j$ is $\Delta W_{ij} = \Phi(S_i \bigcup S_j) - \Phi(S_i)-\Phi(S_j)$. Then, for each pair $(i,j)$ that $\Delta W_{ij}>Th$, we put a triad $(i,j,\Delta W_{ij})$ into max-heap $Max\_H$. In each iteration, we take the triad $(i,j,\Delta W_{ij})$ from the top of $Max\_H$  whose $\Delta W_{ij}$ is the maximum and merge community $i$ and community $j$. For the union operation, we can use the disjoint-set data structure to implement it. Next, we update the graph  $\mathcal{F}$ and put the new triad $(i,k,\Delta W_{ik})$ whose $\Delta W_{ik}>Th$ into $Max\_H$. We continue these steps until $Max\_H$ is empty. The threshold $Th$ controls the size of communities we find. If $Th$ is too high, the detected communities may be too small. If $Th$ is too low, the detected communities may be too big so that even the resolution limit problem will happen. In practice, $\Phi(S)$ often requires a great number of inter-connections for merging two communities. Thus, we need to properly relax $Th$ to be an appropriate small negative value. Such a relaxation will not lead to the resolution limit problem, but can help us find communities with suitable size. To give a theoretical analysis in large network, recalling Theorem 7 introduced in section 3.2.2, the mathematical expression of $Cos(S)$ is given in (10). If we do not want to merge two equal-sized cliques $\mathcal{C}_1$ and $\mathcal{C}_2$ ($|\mathcal{C}_1|=|\mathcal{C}_2| \geq 5$ and there is only one edge between $\mathcal{C}_1$ and $\mathcal{C}_1$), we should have:
$$Th >Cos( \mathcal{C}_1 \bigcup \mathcal{C}_2)- Cos(\mathcal{C}_1)-Cos(\mathcal{C}_2) = -2.947214.$$
We suggest the user to choose an appropriate threshold $Th$ such that $-2.9<Th<0$.

\subsection{Complexity Analysis}
As it has been discussed, the time complexity of Seed Selection is $O(\beta m)$, where $\beta$ is a constant. As for Local Optimization Iteration, we need to go through all the edges in each iteration. Thus, its time complexity is $O(Max\_It \cdot m)$, where $Max\_It$ is the number of iterations that Local Optimization Iteration needs. In practice, we can set $Max\_It = 20$ and it is sufficient for Local Optimization Iteration to converge. For Community Merging, we need to merge two communities and insert new elements into a max-heap in each iteration. We use both path compression and union by size to ensure that the amortized time per union operation is only $O(\alpha(n))$~\cite{tarjan1984worst}~\cite{tarjan1979class}, where $\alpha(n)$ is the inverse Ackermann function which can be viewed as a constant. Thus, we have $O(n)$ for the union operation in general. For insert operation, the worst time complexity is $O(m\log n)$ since we have to go through every neighbor of the community in each iteration. Totally, the time complexity is $O(m+n+m\log n)$.

\section{Experimental Results}
In this section, the proposed algorithm CBCD is compared with the existing state-of-the-art algorithms on both synthetic networks and real networks. Each node in these two different kinds of networks has a ground truth community label. For communities found by the detection algorithms, we need to compare them with the ground truth communities. A criterion is necessary to measure the similarity between the final partition of the algorithm and the actual communities. Many evaluation measures, such as Normalized Mutual Information (NMI)~\cite{danon2005comparing}, ARI~\cite{hubert1985comparing} and Purity, have been proposed to evaluate the clustering quality of algorithms. NMI is the most widely-accepted and important evaluation measure, since it is more discriminatory and more sensitive to errors in the community detection procedure~\cite{danon2005comparing}. We will use NMI as our performance evaluation measure in the experiments. The code of NMI calculation offered by McDaid et al.~\cite{mcdaid2011normalized} can be found at \url{https://github.com/aaronmcdaid/Overlapping-NMI}.

According to the comparative analysis of community detection algorithms~\cite{lancichinetti2009community}, Louvain and Infomap are  two of the best classical algorithms in the literature. Thus, we will take these two algorithms as the competing algorithms. Besides, some novel algorithms proposed in recent years will also be compared with CBCD. These algorithms are listed as follow:
\begin{itemize}
\item Louvain~\cite{blondel2008fast} is a well-known heuristic algorithm based on modularity. The algorithm is composed of two steps which are performed iteratively. The first step is to move each node to the community that the gain of modularity is positive and maximum. The next step is to build a new weighted network whose nodes are communities found in first step. The procedure will continue until modularity achieves a maximum. Louvain can unfold a complete hierarchical community structure for the network. The program can be downloaded from \url{https://sites.google.com/site/findcommunities/}.
\item Infomap~\cite{rosvall2008maps} is another well-known algorithm based on information theory and random walk. It uses the probability flow of random walks taking place over a network as a description of the network structure and decomposes the network into modules using information theoretic result to compress the probability flow. It simplifies the organization of network and highlights their relationships. The program can be downloaded from \url{http://www.mapequation.org/code.html}.
\item FOCS~\cite{bandyopadhyay2015focs} is a fast overlapping community detection algorithm, which can detect overlapped communities using the local connectedness. FOCS takes a parameter OVL as the input, which is a threshold allowing for maximum overlapping between two communities. Since our comparison  is conducted over the non-overlapped algorithms, we set OVL to 0. The program can be downloaded from \url{https://github.com/garishach/focs}.
\item SCD~\cite{prat2014high}~\cite{prat2016put} is a community detection algorithm based on a new community metric WCC. WCC considers the triangle as the basic structure instead of the edge or node. The theoretic analysis given in~\cite{prat2016put} shows that WCC can correctly capture the community structure. The program can be downloaded from \url{https://github.com/DAMA-UPC/SCD}.
\item Attractor~\cite{shao2015community} is an algorithm based on distance dynamics. The fundamental basis of Attractor is to view the whole network as an adaptive dynamical system. Each node in this dynamical system interacts with its neighbors and distances among nodes will be changed  by the interactions. At the same time, distances will affect the interactions conversely. The dynamical system eventually evolves a steady system. The Attractor algorithm require a cohesion parameter $\lambda$ that ranges from 0 to 1, which is used to determine how exclusive neighbors affect distance (positive or negative influence). According to ~\cite{shao2015community}, we set $\lambda=0.5$. The program can be downloaded from  \url{https://github.com/YcheCourseProject/CommunityDetection}.
\end{itemize}

For all experiments, without further statement, we set the threshold parameter of  our algorithm $Th=-2.8$ when $0<|V|<4000$ and $Th=-0.43$ when $4000\leq |V| $, corresponding to small networks and large-scale networks.  All experimental results have been obtained on a workstation with 3.5 GHz Intel(R) Xeon(R) CPU E5-1620 v3 and 16.0 GB RAM. For Louvain, we adopted the lowest partition of the hierarchy, which is stored in the graph.tree file. For Infomap, the number of outer-most loops to run before picking the best solution is specified to be 10.

\begin{figure*}[t]
    \centering
     \centerline{ \includegraphics[width=1.0\linewidth]{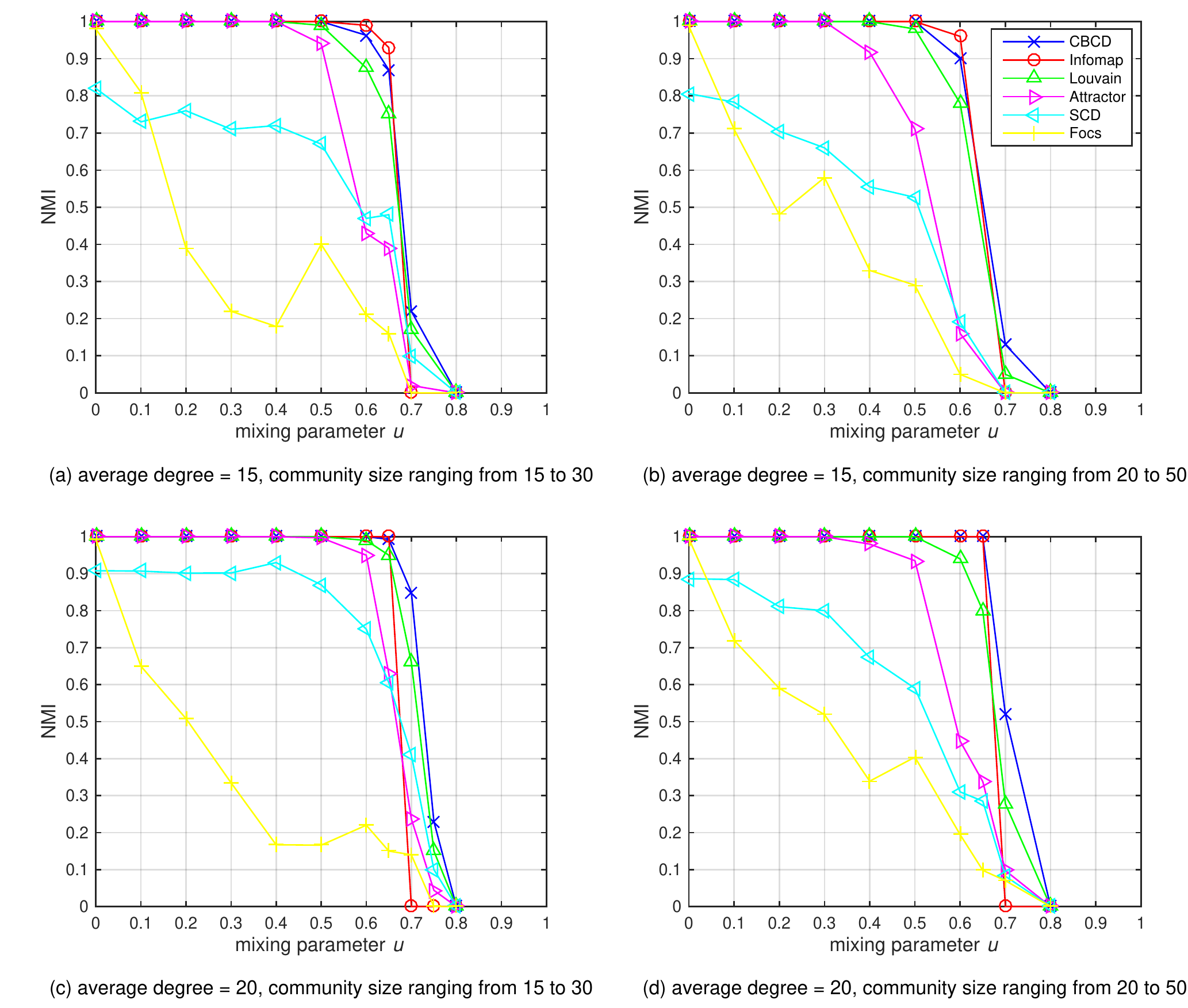}}
	 \caption{ The performance of  different algorithms on LFR benchmark. The panels indicate the NMI value of the detection algorithm as a function of the mixing parameter $u$.}
	  \label{fig6}
\end{figure*}
\begin{figure*}[!htb]
    \centering
     \centerline{ \includegraphics[width=1.0\linewidth]{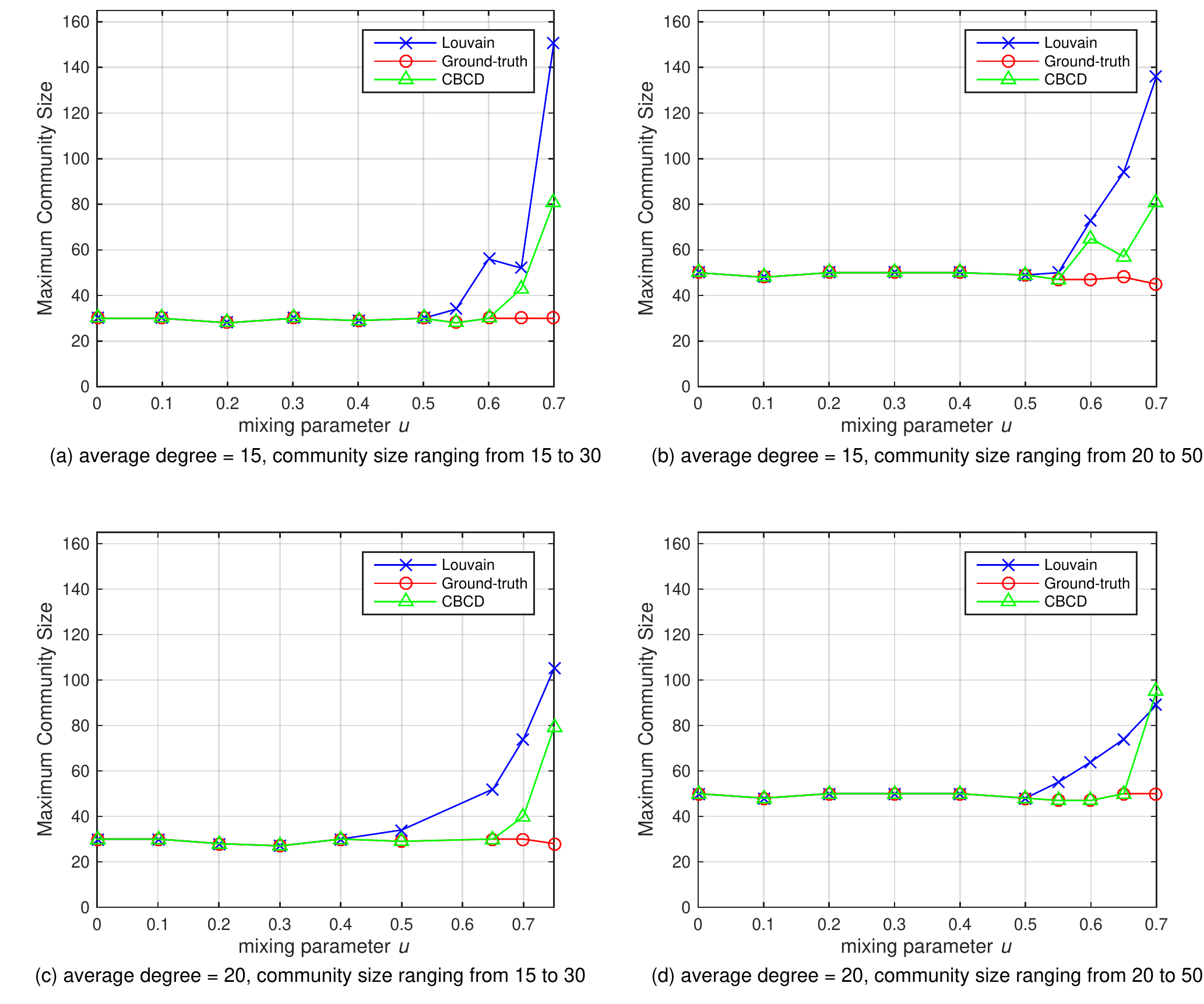}}
	 \caption{ The maximum community size in the partition generated when the mixing parameter $u$ is ranged from 0 to 0.7. The red curve corresponds to the ground-truth partition.}
	  \label{fig7}
\end{figure*}
\subsection{LFR Benchmark}
LFR Benchmark which is introduced by Lancichinetti et al.~\cite{lancichinetti2008benchmark} is a very popular graph simulation model. The most important parameter of LFR Benchmark is the mixing parameter $u$. The community structure of LFR network becomes more pronounced as $u$ decreases. In particular, $u=0$ indicates that each node only connects the nodes inside its community and $u=1$ indicates that each node only connects the nodes outside its community. The program of LFR Benchmark can be downloaded from \url{https://sites.google.com/site/santofortunato/inthepress2}.

We generate several LFR networks characterized by different features to compre the performance of various algorithms. The number of nodes in all LFR networks is fixed to 1000. We consider $2\times2$ cases, in which the community size parameter is specified within the range  [15,30] and [20,50] and the average degree is set to be 15 and  20. For each case, we fix the average degree and the community size, and then increase the mixing parameter $u$ from 0 to 1 to generate a variety of LFR networks for comparison.

The performance comparison result in terms of NMI is shown in Fig. 7. As we can see from Fig. 7, CBCD, Louvain and Infomap almost achieve the best clustering performance. The NMI values of all algorithms will decrease when the mixing parameter $u$ tends to 1. This is because the increment of mixing parameter $u$ will introduce more edges among different communities, making it difficult to identify the underlying true communities. CBCD have better performance when the community size parameter is smaller. Note that even when the community size parameter is relatively big, CBCD still has good  performance in comparison with the other algorithms. Except for Infomap and Focs, the performance of other algorithms become better when the average degree parameter is increased, and CBCD  is better than all other algorithms when the degree is 20. We can find that Infomap always starts to decrease dramatically when $u$ ranges between $0.6$ and $0.7$. It is because that Infomap is based on the random walk dynamics and is  more sensitive to the noisy inter-edges between communities as $u$ tends to 1. By contrast, CBCD is more robust to these noisy inter-edges despite the fact that Infomap is slightly better than CBCD when average degree is 15 and $u$ ranges between $0.55$ and $0.65$. Compared to  Louvain, CBCD is always the winner. Let us consider the maximum community size in the partition obtained by Louvain and CBCD. The relation between the mixing parameter $u$ and  the maximum community size of detected partition is plotted in Fig. 8. The maximum community size of both CBCD and Louvain increases as the mixing parameter $u$ is increased. This demonstrates that the detection algorithms tend to combine of two ground-truth communities when $u$ is high. The maximum community size of CBCD is almost always lower than that of Louvain. We can observe that, especially when the mixing parameter $u$ ranges between $0.5$ and $0.6$, the maximum community size of CBCD is significantly lower than the maximum community size of Louvain. This result show that CBCD can mitigate the resolution problem to some extent.

\subsection{Real-World Network}
\begin{table*}
\centering
\caption{The characteristics of real-world network data sets.}
\begin{tabular}{llllll} 
\toprule
  Data Sets  
   &    $|V|$ &  $|E|$ & $\langle d \rangle$ & $d_{max}$ & $|C|$ \\
  \midrule
  football&  $115$ & $613$  & $10.57$  & $12$  & $12$ \\
  karate &   $34$  & $78$   & $4.59$   & $17$  & $2$ \\
  personal&  $561$ & $8375$ & $29.91$  & $166$ & $8$ \\
  polbooks&  $105$ & $441$  & $8.4$    & $25$  & $3$ \\
  polblogs&  $1490$& $19090$& $27.32$  & $351$ & $2$  \\
  \bottomrule
  \end{tabular}
  \label{tbl:table3}
\end{table*}
\begin{table*}
\centering
\caption{The performance of different  algorithms on the real-world network.}
\begin{tabular}{l|lllll|lllll} 
\toprule

      & \multicolumn{5}{|c|}{NMI}  & \multicolumn{5}{c}{NC}  \\
  \midrule
 &football  &  karate & personal & polbooks &  polblogs&football  &  karate & personal & polbooks &  polblogs\\

 CBCD&  $0.734$ & \textbf{0.837}  & \textbf{0.3639}  & \textbf{0.330}  & \textbf{0.391} & $9$  & $2$   & $11$  & $4$  & $6$\\
  Infomap &   $0.833$  & $0.563$   & $0.248$   & $0.293$  & $0.268$ & $12$ & $3$ & $6$ & $5$ & $303$ \\
  Louvain& $0.838$ & $0.259$ & $0.08$  & $0.158$ & $0.212$& $12$ & $7$ & $17$ & $10$ & $32$ \\
  SCD&  \textbf{0.840} & $0.395$& $0.179$  & $0.116$ & $0.085$ & 14& 8 & 125 & 23 & 664 \\
  Focs&  $0.392$& $0.189$& $0.171$  & $0.166$ & $0.106$ & 5  & 1  & 13 & 9 & 19\\
  Attractor& $0.833$& $0.04$& $0.299$  & $0.315$ & $0.124$ &12 &1  &56  &7  &313\\
\bottomrule
  \end{tabular}
  \label{tbl:table5}
\end{table*}
In most real-world networks, each node has no ground-truth label.  The modularity is typically used to evaluate the quality of detected communities. However, as we have discussed before, modularity is not a good quality measure of communities because of the resolution limit problem.  Besides, modularity is found out owning the tendency of following the same general pattern for different classes of  networks \cite{leskovec2010empirical}. Thus, we only conduct our experiment on several well-known real-world networks with ground-truth communities: Karate (karate) \cite{zachary1977information}, Football (football) \cite{girvan2002community}, Personal Facebook (personal) \cite{wilson2014testing}, Political blogs (polblogs) \cite{adamic2005political}, Books about US politics (polbooks) \cite{krebs2013social}. The detailed statistics of real-world networks are given in Table 3, where $|V|$ denotes the number of the nodes, $|E|$ denotes the number of the edges, $d_{max}$ denotes the maximal degree of the nodes, $\langle d \rangle$ denotes the average degree of the nodes and $|C|$ denotes the number of ground-truth communities in the network. The performance of the algorithms on the real-world network is shown in Table 4, where NC is the number of  communities detected by the algorithm.

\textbf{American college football}:  American college football network describes football games between Division IA colleges during the regular season in Fall 2000. It has 115 teams and 631 games between these teams. For each node (team), there is an edge connecting two nodes if two teams played a game. The teams were partitioned into 12 conferences (communities). Louvain, Infomap, SCD and Attractor all have good performance on the football data set, and SCD achieve the best performance among these algorithms. We have to admit the fact that these algorithms outperform our method on the football data set. Fig. 9 plots the variation of NMI of the partition detected by CBCD when the threshold $Th$ used in the merging operation increases from -2.8 to 0. As  $Th$ increases from -2.8 to -2.2, NMI increases to 0.773 which is the maximal NMI value of CBCD on the football data set. NMI begins to decrease when  $Th$ further increases.  Note that the number of detected communities increases with the increment of $Th$. The number of detected communities is 10 when $Th=-2.2$. The quality of algorithm 4 (Community Merging) is mainly determined by the output of algorithm 3 (Local Optimization Iteration). Although CBCD can achieve a NMI value of 0.773, it still cannot beat other algorithms except Focs. This demonstrates that  there is still room for improving  algorithm 3.
\begin{figure}[!ht]
    \centering
       \includegraphics[width=1\linewidth]{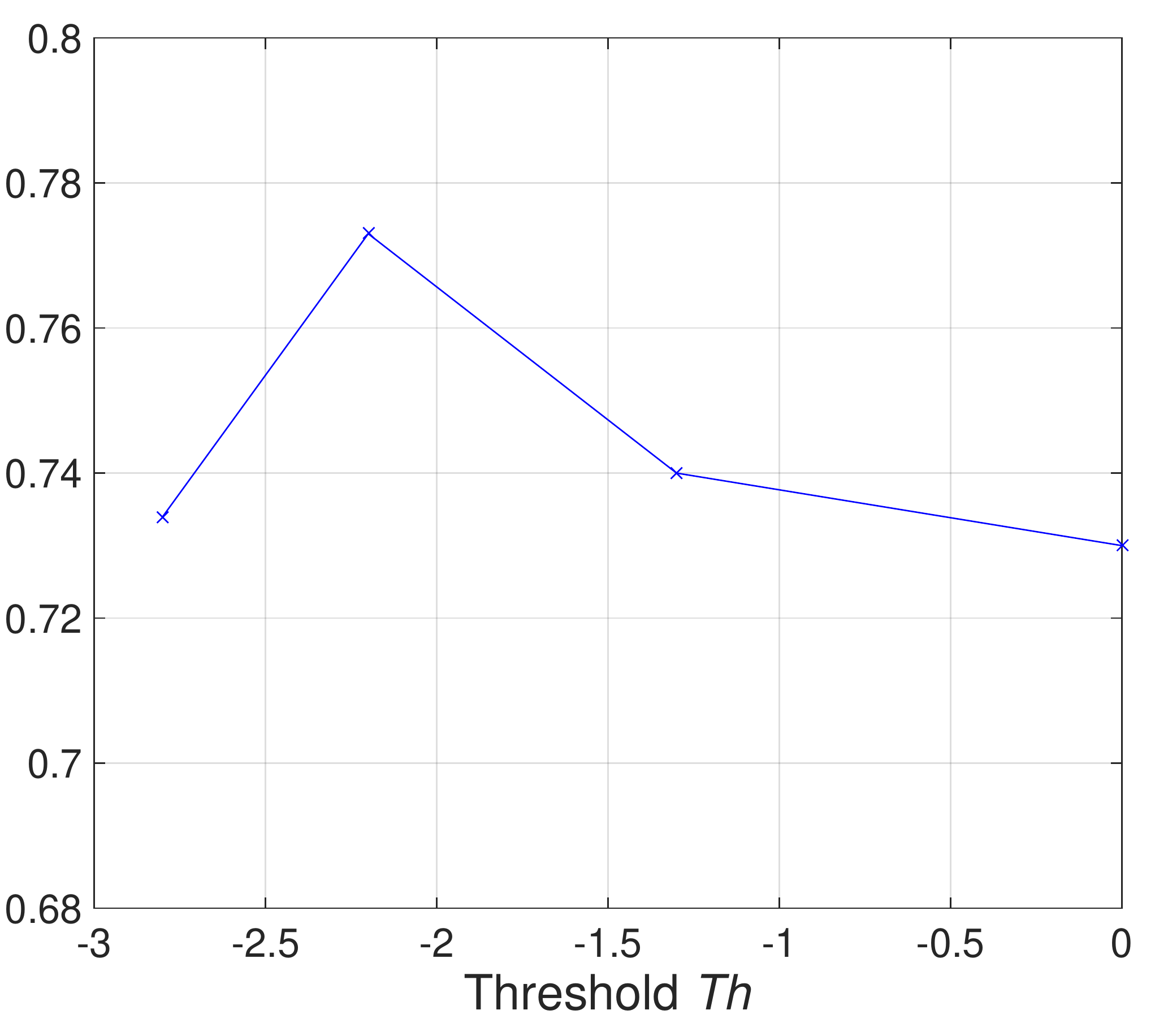}
	\centering	
 	\caption{ The variation of NMI of the partition detected by CBCD when the threshold $Th$ is increased from -2.8 to 0 on the football network. }
	  \label{fig8}
\end{figure}

\textbf{Zachary's karate club network}: Karate is a famous network derived  from the Zachary's observation about a karate club. The network describes  the friendship among the members of a karate club. The network was divided into two parts because of the divergence between administrator and instructor. According to Table 4, CBCD outperforms all other algorithms and achieves a NMI value of 0.840. Two communities are successfully found by CBCD, but our algorithm classifies one node `10' into the wrong community. We observe that this node only have two edges and each edge connects one of two communities respectively. It is difficult to decide which community it really belongs to. Thus, we consider this node as a noisy node. In fact, it can make sense to assign node `10' to both communities in the context of overlapping community detection. However, this study mainly focuses on  non-overlapping community detection. If we delete node '10' from the karate, our CBCD algorithm can achieve the perfect performance of NMI=1. Its output exactly matches the partition of  ground-truth  communities. Infomap achieves the second best performance. For Attractor, the worst performance is due to that it puts  all the members into one community.

\textbf{Personal Facebook network}: Personal is the network which gives the friendship structure of the first author, where each individual (node) is labeled according to the time period when he or she met the first author. Persons are divided into the different groups according to their locations. CBCD achieves the best performance with relatively high quality (NMI=0.3639) on the personal data set. For Attractor and Infomap, they also achieve good performance. Louvain achieves the worst result.

\textbf{Books about US politics}: This network consists of 105 nodes and 441 edges, which is derived from the politic books about US politics published in 2004 when presidential election takes place. Each node represents the book sold at \textit{Amazon.com}, and each edge represents that two books are frequently  co-purchased by the same buyer.  Each book is labeled with "liberal", "neutral" or "conservative", that is given by \textit{Amazon.com}. CBCD gives the best partition with NMI = 0.33 among the comparing algorithms. For Attractor and Infomap yield good performance while Lovain, SCD and Focs have relatively bad performance.

\textbf{U.S. political blog}: The polblog network consists of 1490 nodes and 19090 edges, which describes the degree of interaction between liberal and conservative blogs. Compared to other algorithms, CBCD has the best performance on the polblog data set. Infomap and Louvain also produce  good partitions. Attractor, Focs, and SCD yield relatively bad partitions.
\subsection{Large-Scale Real-World  Network}
\begin{figure*}[!t]
    \centering
     \centerline{ \includegraphics[width=1.0\linewidth]{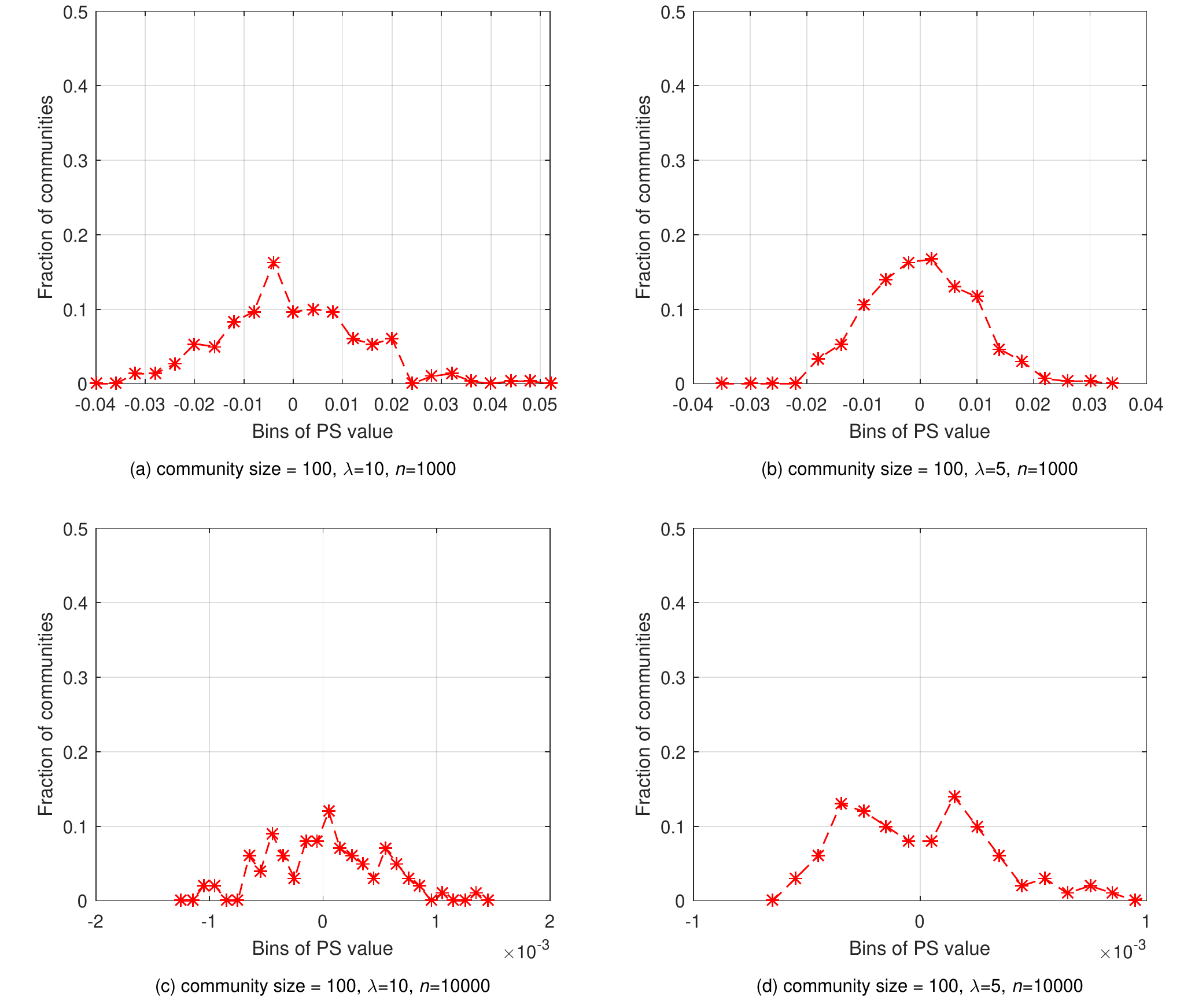}}
	 \caption{ The distribution of the PS values of a community for E-R random networks with different parameters.}
	  \label{fig10}
\end{figure*}
The large-scale real networks which are provided from \cite{yang2015defining} all have overlapping ground-truth communities.  Despite it is out of the scope of our paper,  we will still run CBCD along with other algorithms on these networks to test the performance of our algorithm on large-scale real networks.  To evaluate the performance of community detection algorithms on the networks with overlapping community structures, Overlapping Normalized Mutual Information (ONMI) \cite{lancichinetti2009detecting} is the major evaluation metric in this section. Since we will make a comparison on the networks with overlapping structures, then we set OVL of Focs to 0.6 for detecting overlapping communities. We choose two large-scale networks, Amazon and DBLP,  for testing the performance of different methods. The specific information of these two networks is given as follow.

\textbf{Amazon}:  It is an undirected network collected by crawling the Amazon website, where each node is a product sold on the website and  an edge exists  between two nodes (products) if they are frequently co-purchased.  The ground-truth communities are determined by the product category. Each connected component in a product category is regarded as an independent ground-truth community. The whole network have 334863 nodes, 925872 edges and 70928 communities. Ninety-one percent of the nodes participate in at least two communities.

\textbf{DBLP}:  It is a co-authorship network derived from the DBLP computer science bibliography.  Each author of a paper is viewed as a node. Two authors are connected by an edge if they have published at least one paper together.  Publication venue, e.g, journal or conference,  is the indicator of ground-truth community.  The authors who publish papers on the same journal or conference form a community. The whole network have 317080 nodes, 1049866 edges and 13477 communities. Thirty-five percent of the nodes participate in at least two communities.
\begin{table}[htb]
\small
\centering
\setlength{\abovecaptionskip}{0pt}
\setlength{\belowcaptionskip}{10pt}
\caption{The performance comparison of different algorithms on large-scale real networks with overlapping ground-truth communities}
\begin{tabular}{cm{12mm}<{\centering}m{8mm}<{\centering}m{6mm}<{\centering}m{8mm}<{\centering}m{6mm}<{\centering}}\hline
Data sets  & Algorithm & ONMI & NC &  ET   \\ \hline
\multirow{5}*{DBLP}     & CBCD     & $0.132$         &  40k  & 48s   \\
                        & Infomap    & $0.008$          &  109k   & 120s      \\
                        & Louvain   & $0.124$         & 170k  & 13s       \\
                        & SCD     & $0.146$    &  140k  & 15s \\
			& Focs     & $\textbf{0.213}$    &  24k  & 7s     \\
			& Attractor     & $0.061$    &  17k  & 43min  \\ \hline
\multirow{5}*{Amazon}    & CBCD     & $\textbf{0.246}$   &    40k   & 42s \\
                        & Infomap    &  $0.057$           &  213k    & 132s     \\
                        & Louvain   & $0.154$         & 266k  &  16s           \\
                        & SCD     & $0.158$      & 141k  & 6s    \\
			& Focs     & $0.207$    &  20k  & 5s      \\
			& Attractor     & $0.201$    &  23k  & 22min\\ \hline
\end{tabular}
\end{table}

Table 4 summaries the experimental results of different algorithms on two large-scale real networks, where NC is the number of detected communities and ET is the execution time of various algorithms.  For the DBLP network, Focs achieves the best performance and SCD achieves the second-best performance. Despite the performance of CBCD on DBLP is not as good as these two algorithms, CBCD is still better than others.  For the Amazon netowrk, CBCD outperforms the other algorithms. Focs is the second-best performer and Attractor is slightly inferior to Focs.  Infomap has the worst performance on both DBLP and Amazon,  which is in contrast with its performance on small real networks and LFR networks.
Although CBCD is effective on  detecting meaningful  communities, it has no obvious advantage with respect to the  execution time. This is what we should focus on in our future work.
\subsection{The Distribution of PS Values}
In this section,  we study the distribution of PS values defined in Formula (6)  of a community under the E-R model. First,  we generate 300 random networks for the specific parameters under the E-R model. These parameters are the average degree $\lambda$ and the  network size $n$ of random network.  Then,  we calculate the PS value of  a given community $S$, which is composed of 100 fixed nodes.  Consequently,  we obtain 300 PS values of  the given communities derived from 300 different random networks.  We divide PS values into many bins with equal length, where the low (high)  bins correspond to the set of  lower (higher) PS values.   In Fig. 10, bins are plotted  on the $x$-axis,  and for each bin,   the fraction of communities  whose PS  values  fall into that bin are plotted on the $y$-axis.  We can observe that the distribution of  PS value follows a Gaussian-like distribution.  The PS metric values are concentrated near 0 and most values fall into a very narrow interval. This phenomenon  corresponds to our theoretical result in section 3.2.1.  Comparing Fig. 10 (a) with  Fig. 10 (b) and comparing Fig. 10(c) with Fig. 10(d),  we can find that the interval that most PS  values fall into  becomes  shorter with the decrease of average degree $\lambda$. Besides,  the interval that most PS values fall into sharply shortens when the network size $n$ increases. It is a remarkable fact that the PS value of a community in the random networks generated from the E-R model is a very small value. This demonstrates that the PS value of a community is an effective metric for quantifying the goodness of a  community structure.
\section{Conclution}
In this paper, we introduce two novel node-centric community evaluation functions by connecting correlation analysis with community detection. We further show that  the correlation analysis is a novel theoretical framework which unifies some existing quality functions and converts community detection into a correlation-based optimization problem.  In this framework, we choose PS-metric and $\phi$-coefficient to eliminate the influence of random fluctuations and mitigate the resolution limit problem. Furthermore, we introduce three key properties used in mining association rule into the context of community detection to help us choose the appropriate correlation function.  A correlation-based community detection algorithm CBCD that makes use of PS-metric and $\phi$-coefficient  is proposed in this paper. Our proposed algorithm outperforms five existing state-of-the-art algorithms on both LFR benchmark networks and real-world networks. In
the future, we will investigate more correlation functions and extend our method  to overlapping community detection.



%




\ifCLASSOPTIONcompsoc
  \section*{Acknowledgments}
\else
  \section*{Acknowledgment}
\fi

This work was partially supported by the Natural Science Foundation of China under Grant No. 61572094.

\ifCLASSOPTIONcaptionsoff
  \newpage
\fi



\bibliographystyle{IEEEtran}
\bibliography{IEEEtran}
\end{document}